\documentclass[a4paper,fleqn]{cas-sc}

\usepackage[authoryear,longnamesfirst]{natbib}

\def\tsc#1{\csdef{#1}{\textsc{\lowercase{#1}}\xspace}}
\tsc{WGM}
\tsc{QE}

\usepackage{amssymb,amsthm,amsmath,amsfonts}
\usepackage{parskip}
\usepackage{xspace}
\usepackage[dvipsnames]{xcolor}
\usepackage{enumitem}
\usepackage[hyphens]{xurl}
\usepackage{doi}
\usepackage[utf8x]{inputenc}
\usepackage{dsfont}
\usepackage{longtable}
\usepackage{bold-extra}
\usepackage{thm-restate}
\usepackage{mathdots}
\usepackage{makecell}
\usepackage[T1]{fontenc}
\usepackage{mathrsfs}
\usepackage{listings}
\usepackage[official]{eurosym}
\usepackage{colophon}
\usepackage{multicol}

\renewcommand{\eqref}[1]{(\ref{#1})}

\newtheorem{theorem}{Theorem}
\newtheorem{lemma}[theorem]{Lemma}

\newtheorem{corollary}[theorem]{Corollary}
\newdefinition{definition}[theorem]{Definition}
\newdefinition{remark}[theorem]{Remark}

\usepackage{algorithm}
\usepackage[noend]{algorithmic}

\newcommand{\on}[1]{\ensuremath{\operatorname{#1}}}
\newcommand{\abs}[1]{\ensuremath{\left\lvert #1 \right\rvert}}
\newcommand{\smallabs}[1]{\ensuremath{\lvert #1 \rvert}}
\newcommand{\eps}{\ensuremath{\varepsilon}}

\newcommand{\ZZ}{\ensuremath{\mathbb{Z}}}

\newcommand{\RR}{\ensuremath{\mathbb{R}}}

\newcommand{\INSTANCES}[1]{\ensuremath{\mathcal{{I}}_{#1}}\xspace}
\newcommand{\OPT}{\ensuremath{\on{\textsc{Opt}}}\xspace}
\newcommand{\ALG}{\ensuremath{\on{\textsc{Alg}}}\xspace}
\newcommand{\PAR}{\ensuremath{\on{\textsc{PO}}}\xspace}

\newcommand{\ALPHA}{\ensuremath{\alpha}} 
\newcommand{\BETA}{\ensuremath{\beta}} 
\newcommand{\GAMMA}{\ensuremath{\gamma}} 
\newcommand{\ABC}{\ensuremath{(\ALPHA,\BETA,\GAMMA)}\xspace}

\makeatletter
\newcommand{\myoverset}[3][-0.3ex]{%
  \mathrel{\mathop{#3}\limits^{
    \vbox to#1{\kern-2\ex@
    \hbox{$\scriptstyle#2$}\vss}}}}
\makeatother
\newcommand{\orarrow}{\ensuremath{\myoverset{r}{\rightarrow}}}

\definecolor{custom_yellow}{RGB}{246,219,166}
\definecolor{custom_yellow_dark}{RGB}{196,136,19}
\definecolor{custom_blue}{RGB}{166,193,246}
\definecolor{custom_blue_dark}{RGB}{19,79,196}
\definecolor{custom_red}{RGB}{244,123,126}
\definecolor{custom_red_dark}{RGB}{220,18,23}

\newcommand{\MAXDEGREE}{\ensuremath{\Delta}\xspace}

\newcommand{\AT}{\ensuremath{\kappa}} 

\newcommand{\ALGS}[1]{\ensuremath{\mathcal{A}_{#1}}\xspace}
\newcommand{\ORALG}[1]{\ensuremath{{#1}_{\textsc{a}}}\xspace}
\newcommand{\ORTRANS}[1]{\ensuremath{{#1}_{\textsc{r}}}\xspace}
\newcommand{\ONLRED}[1]{\ensuremath{(\ORALG{#1},\ORTRANS{#1})}\xspace}

\newcommand{\CCWM}[2]{\ensuremath{\mathcal{C}_{#1}^{#2}}\xspace}

\newcommand{\ASG}[1]{\ensuremath{\textsc{ASG}_{#1}}\xspace}
\newcommand{\ASGINFTY}{\ensuremath{\textsc{ASG}}\xspace}
\newcommand{\BDVC}[1]{\ensuremath{\textsc{VC}_{#1}}\xspace}
\newcommand{\VC}{\ensuremath{\textsc{VC}}\xspace}
\newcommand{\COL}[1]{\ensuremath{#1\textsc{-Spill}}\xspace}
\newcommand{\INTER}[1]{\ensuremath{\textsc{IR}_{#1}}\xspace}
\newcommand{\IR}{\ensuremath{\textsc{IR}}\xspace}
\newcommand{\SAT}[1]{\ensuremath{#1\textsc{-SatD}}\xspace}
\newcommand{\DOM}{\ensuremath{\textsc{Dom}}\xspace}
\newcommand{\PAG}[1]{\ensuremath{\textsc{Pag}_{#1}}\xspace}
\newcommand{\ETA}{\ensuremath{\eta}}
\newcommand{\ETAZERO}{\ensuremath{\ETA_0}}
\newcommand{\ETAONE}{\ensuremath{\ETA_1}}
\newcommand{\PAIRETAFORCC}{\ensuremath{\ETAZERO,\ETAONE}}
\newcommand{\PAIRETA}{\ensuremath{(\PAIRETAFORCC)}\xspace}

\newcommand{\PHI}{\ensuremath{\varphi}}
\newcommand{\PHIZERO}{\ensuremath{\PHI_0}}
\newcommand{\PHIONE}{\ensuremath{\PHI_1}}
\newcommand{\PAIRPHIFORCC}{\ensuremath{\PHIZERO,\PHIONE}}
\newcommand{\PAIRPHI}{\ensuremath{(\PAIRPHIFORCC)}\xspace}
\newcommand{\MU}{\ensuremath{\mu}}
\newcommand{\MUZERO}{\ensuremath{\MU_0}}
\newcommand{\MUONE}{\ensuremath{\MU_1}}
\newcommand{\PAIRMUFORCC}{\ensuremath{\MUZERO,\MUONE}}
\newcommand{\PAIRMU}{\ensuremath{(\PAIRMUFORCC)}\xspace}
\newcommand{\ZM}{\ensuremath{Z}}
\newcommand{\ZMZERO}{\ensuremath{\ZM_0}}
\newcommand{\ZMONE}{\ensuremath{\ZM_1}}
\newcommand{\PAIRZMFORCC}{\ZMZERO,\ZMONE}
\newcommand{\PAIRZM}{(\PAIRZMFORCC)\xspace}

\newcommand{\MINDEGREE}{\ensuremath{\delta}\xspace}

\newcommand{\BDIS}[1]{\ensuremath{\textsc{IS}_{#1}}\xspace}
\newcommand{\IS}{\ensuremath{\textsc{IS}}\xspace}
\newcommand{\SCH}[1]{\ensuremath{\textsc{Sch}_{#1}}\xspace}
\newcommand{\SCHEDULING}{\ensuremath{\textsc{Sch}}\xspace}
\newcommand{\CLI}[1]{\ensuremath{\textsc{Cli}_{#1}}\xspace}
\newcommand{\CLIQUE}{\ensuremath{\textsc{Cli}}\xspace}
\newcommand{\SP}[1]{\ensuremath{\textsc{SP}_{#1}}\xspace}
\newcommand{\SETPACKING}{\ensuremath{\textsc{SP}}\xspace}
\newcommand{\MCS}[2]{\ensuremath{\textsc{M}#1\textsc{CS}_{#2}}}
\newcommand{\MAT}[1]{\ensuremath{\textsc{MM}_{#1}}\xspace}
\newcommand{\MATCHING}{\ensuremath{\textsc{MM}}\xspace}

\usepackage{tikz}
\usepackage{graphicx}
\usetikzlibrary{arrows.meta}
\usetikzlibrary{patterns,patterns.meta}
\usepackage{chemarrow}
\tikzset{
	vertex/.style = {
		shape = circle,
		draw = black,
		fill = gray!15,
		minimum size = 1cm,
		scale = 0.8
	},
	nnvertex/.style = {
		shape = circle,
		draw = black,
		minimum size = 0.2cm
	},
	vertexwcolor/.style = {
		shape = circle,
		draw = black,
		minimum size = 1cm,
		scale = 0.8
	},
	strictorarrow/.style = {
		-{Latex[line width=0.5pt, fill=white,scale = 1.5]}
	},
	orarrow/.style = {
		-{Latex[scale = 1.5]}
	},
	biorarrow/.style = {
		{Latex[scale = 1.5]}-{Latex[scale = 1.5]}
	},
	polyominofill/.style={
  		shape = rectangle,
		minimum size = 0.6cm,
		draw=black,
		very thick,
		fill=gray!15
  	},
	polyominodarkfill/.style={
		shape = rectangle,
		minimum size = 0.6cm,
		draw=black,
		very thick,
		fill=gray!60
  	},
	polyominofill2/.style={
		shape = rectangle,
		minimum size = 0.75cm,
		draw=black,
		very thick,
		fill=gray!15
  	},
	dummyfill/.style={
		shape = rectangle,
		minimum size = 0.75cm,
		draw=black,
		very thick,
		fill=gray!60
	},
	bfstreefill/.style={
		shape = circle,
		text width = 0.3cm,
		draw=black,
		thick,
		fill=gray!15
	},
	prompt/.style = {
		scale = 0.8,
		anchor = west
	}
}

\usepackage{hyperref}
\hypersetup{ 
	colorlinks = true,
    linkcolor={black},
    citecolor={blue!60!black},
    urlcolor={red!60!black}
}

\begin{document}
\let\WriteBookmarks\relax
\def\floatpagepagefraction{1}
\def\textpagefraction{.001}

\shorttitle{Comparing the Hardness of Online Minimization and Maximization Problems with Predictions}    

\shortauthors{M. Berg}  

\title [mode = title]{Comparing the Hardness of Online Minimization and Maximization Problems with Predictions}  

\tnotemark[1] 

\tnotetext[1]{That author of this paper was supported in part by the Independent Research Fund Denmark, Natural Sciences, grant DFF-4283-00079B and in part by the Innovation Fund Denmark, grant 9142-00001B, Digital Research Centre Denmark, project P40: Online Algorithms with Predictions. An extended abstract of this paper was published in the International Joint Conference on Theoretical Computer Science -- Frontier of Algorithmic Wisdom (IJTCS-FAW), volume 15828 of Lecture Notes in Computer Science, pages 33--48. Springer, 2025.} 

%

\author[1]{Magnus Berg}[orcid=0000-0001-8637-7113]

\cormark[1]


\ead{mbmo@iba.dk}

\ead[url]{https://orcid.org/my-orcid?orcid=0000-0001-8637-7113}


\affiliation[1]{organization={IBA International Business Academy},
            addressline={Havneparken 1}, 
            postcode={6000}, 
            city={Kolding},
            country={Denmark}}

\cortext[1]{Corresponding author}


\begin{abstract} 
We build on the work of Berg, Boyar, Favrholdt, and Larsen, who developed a complexity theory for online problems with and without predictions (IJTCS-FAW, volume 15828 of LNCS, Springer, 2025) where they define a hierarchy of complexity classes that classifies online problems based on the competitiveness of best possible deterministic online algorithms for each problem.
Their work focused on online minimization problems and we continue their work by considering online maximization problems.

We compare the competitiveness of the base online minimization problem from Berg, Boyar, Favrholdt, and Larsen, Asymmetric String Guessing, to the competitiveness of Online Bounded Degree Independent Set.
Formally, we show that there exists algorithms of any given competitiveness for Asymmetric String Guessing if and only if there exists algorithms of the same competitiveness for Online Bounded Degree Independent Set, while respecting that the competitiveness of algorithms is measured differently for minimization and maximization problems.
Beyond this, we give several hardness preserving reductions between different online maximization problems, which imply new membership, hardness, and completeness results for the complexity classes.
Finally, we show new positive and negative algorithmic results for (among others) Online Bounded Degree Independent Set, Online Interval Scheduling, Online Set Packing, and Online Bounded Degree Clique.
\end{abstract}


\maketitle


\section{Introduction} 

Recently, Berg, Boyar, Favrholdt, and Larsen introduced a complexity theory for online problems with (and without) predictions~\cite{BBFL25}, where they classify several hard online minimization problems based on the competitiveness of best possible online algorithms for each problem. 
In~\cite{BBFL25}, the authors only considers online minimization problems.
In this work, we continue their work by comparing the hardness of various online maximization problems to each other, and by comparing the hardness of online minimization problems to the hardness of online maximization problems.
Our main results concern proving membership, hardness and completeness results for online maximization problems for their complexity classes.
The main challenge, and thus also contribution, of this paper is a direct comparison between the competitiveness of algorithms for an online variant of Independent Set and an online variant Vertex Cover.
Contrary to the offline setting, where it is known that approximating Vertex Cover is much easier than approximating Independent Set, we show that in the online setting, and in the online setting with binary predictions, the existence of algorithms of any competitiveness for Independent Set implies the existence of equally good algorithms for Vertex Cover and vice versa.

An \emph{online problem} is an optimization problem, where the input is divided into a sequence of \emph{requests} that are given one-by-one.
When an online algorithm receives a request, it must make an irrevocable decision about the request before receiving the next request.
We analyze the performance of algorithms using competitive analysis~\cite{ST85,BE98,K16}; the most standard framework for online algorithms.
Competitive analysis is a framework for worst-case guarantees, where we measure the quality of an algorithm by comparing its performance over all possible input sequences to the performance of an optimal offline algorithm.

In recent years, there has been an increasing interest in algorithms with predictions~\cite{ALPS}, not least in the context of online algorithms~\cite{LV21,KPS18}.
Here, algorithms are endowed with a predictor to help them create their solution to a given instance.
In the perfect world, online algorithms perform optimally when given perfect predictions and no worse than the best known purely online algorithm when the predictions are erroneous or adversarial.
This ideal case is often not realizable and instead we determine an error measure through which we evaluate the quality of the predictions, and then we express the competitiveness of the algorithms as a function of the error.

The recently developed complexity theory for online problems with predictions~\cite{BBFL25} allows for classifying online problems based on their \emph{hardness}, which is given by the competitiveness of best possible deterministic online algorithms for each problem.
Formally, for any $t \in \ZZ^+ \cup \{\infty\}$ and any pair of error measures $\PAIRETA$, the complexity class $\CCWM{\PAIRETAFORCC}{t}$ is defined as the closure of \emph{$(1,t)$-Asymmetric String Guessing with Unknown History and Predictions} ($\ASG{t}$), under the \emph{as hard as} relation (see Definition~\ref{def:as_hard_as}) with respect to $\PAIRETA$~\cite{BBFL25}.
We give a brief review of the central components of the complexity classes in Section~\ref{sec:complexity_recap}.

\subsection{Our Contribution}

Our contribution revolves around proving several membership, hardness, and completeness results for the complexity classes from~\cite{BBFL25}, focusing on  online maximization problems.
We start by considering the dual problem of the $\CCWM{\PAIRETAFORCC}{t}$-complete problem Online $t$-Bounded Degree Vertex Cover with Predictions ($\BDVC{t}$), namely Online $t$-Bounded Degree Independent Set with Predictions ($\BDIS{t}$).
From this problem, we show a number of reductions proving hardness relations between $\BDIS{t}$ and other relevant online maximization problems such as online variants of \emph{Interval Scheduling}, \emph{Set Packing} etc.

After comparing the hardness of relevant online maximization problems, we relate ourx results to the existing complexity classes from~\cite{BBFL25}, by proving that $\BDIS{t}$ is equivalent to $\ASG{t}$, and thus also to $\BDVC{t}$, in terms of hardness with respect to the canonical pair of error measures $\PAIRMU$.
We do this by showing that there exists an algorithm for $\BDIS{t}$ of any competitiveness if and only if there exists an algorithm of the exact same competitiveness for $\ASG{t}$, while respecting that the competitiveness of algorithms for online maximization problems is measured differently from the competitiveness of algorithms for online minimization problems (see Definition~\ref{def:competitiveness}).

Having established this relation between $\ASG{t}$ and $\BDIS{t}$, we conclude that $\BDIS{t}$ is $\CCWM{\PAIRMUFORCC}{t}$-complete.
This result in turn implies that the online maximization problems we had earlier compared to $\BDIS{t}$ are now members, hard, or complete for $\CCWM{\PAIRMUFORCC}{t}$, depending on their relation to $\BDIS{t}$.
Finally, we extend several positive algorithmic results from~\cite{BBFL25} to members and complete problems, and strong negative algorithmic results from~\cite{ABEFHLPS23} to all complete and hard problems.


\begin{figure}
\centering
\begin{tikzpicture}
\def\ydist{1.5}
\def\xdist{1}

\def\ymargin{0.5}
\def\xmargin{1.25}

\def\ylevone{0}
\def\ylevzero{\ylevone - \ydist}
\def\ylevtwo{\ylevone + \ydist}
\def\ylevthree{\ylevtwo + \ydist}
\def\ylevfour{\ylevthree + \ydist}
\def\ylevfive{\ylevfour + \ydist}
\def\ylevsix{\ylevfive + \ydist}

\def\xminlevone{0}
\def\xminlevtwo{\xminlevone + \xdist}
\def\xminlevthree{\xminlevtwo + \xdist}
\def\xminlevfour{\xminlevthree + \xdist}
\def\xminlevfive{\xminlevfour + \xdist}

\def\xmaxlevone{\xminlevfive + 2.5*\xmargin}
\def\xmaxlevtwo{\xmaxlevone + \xdist}
\def\xmaxlevthree{\xmaxlevtwo + \xdist}
\def\xmaxlevfour{\xmaxlevthree}
\def\xmaxlevfive{\xmaxlevfour + \xdist}

%
%
%

\node at (\xminlevthree,\ylevsix + 0.5*\ydist) {\textbf{Minimization Problems}};

\node[gray!50] (asgt) at (\xminlevfour,\ylevtwo) {$\ASG{t}$};
\node[gray!50] (asg) at (\xminlevthree,\ylevsix) {$\ASGINFTY$};

\node[gray!50] (irt) at (\xminlevtwo ,\ylevtwo) {$\INTER{t}$};
\node[gray!50] (ir) at (\xminlevone,\ylevfour) {$\IR$};

\node[gray!50] (vct) at (\xminlevthree,\ylevthree) {$\BDVC{t}$};
\node[gray!50] (vc) at (\xminlevthree,\ylevfour) {$\VC$};

\node[gray!50] (pagt) at (\xminlevfour,\ylevone) {$\PAG{t}$};

\node[gray!50] (kspill) at (\xminlevfive,\ylevfour) {$\COL{k}_{t+k+1}$};

\node[gray!50] (2satD) at (\xminlevone,\ylevfive) {$\SAT{2}$};

\node[gray!50] (dom) at (\xminlevthree,\ylevfive) {$\DOM$};


\draw[biorarrow,gray!50] (irt) -- (vct);
\draw[biorarrow,gray!50] (asgt) -- (vct);
\draw[biorarrow,gray!50] (irt) -- (asgt);
\draw[strictorarrow,gray!50] (vct) -- (vc);
\draw[orarrow,gray!50,dashed] (vc) -- (dom);
\draw[strictorarrow,gray!50] (irt) -- (ir);
\draw[orarrow,gray!50] (ir) -- (vc);
\draw[orarrow,gray!50] (ir) -- (2satD);
\draw[strictorarrow,gray!50] (vc) to[out = 130, in = -130] (asg);
\draw[strictorarrow,gray!50] (dom) -- (asg);
\draw[orarrow,gray!50] (asgt) -- (kspill);
\draw[strictorarrow,gray!50] (pagt) -- (asgt);

%

\draw[dashed] (\xminlevfive + \xmargin,\ylevone-\ymargin) -- (\xminlevfive + \xmargin,\ylevsix + \ymargin);

\node at (\xmaxlevtwo + 0.25*\xdist,\ylevsix + 0.5*\ydist) {\textbf{Maximization Problems}};

\node (ist) at (\xmaxlevone,\ylevtwo) {$\BDIS{t}$};
\node (is) at (\xmaxlevone,\ylevfive) {$\IS$};

\node (clit) at (\xmaxlevfour,\ylevthree) {$\CLI{t}$};
\node (cli) at (\xmaxlevfour,\ylevfive) {$\CLIQUE$};

\node (scht) at (\xmaxlevone,\ylevthree) {$\SCH{t}$};
\node (sch) at (\xmaxlevone, \ylevfour) {$\SCHEDULING$};

\node (spt) at (\xmaxlevfour,\ylevtwo) {$\SP{t}$};
\node (sp) at (\xmaxlevfour,\ylevfour) {$\SETPACKING$};

\node (matt) at (\xmaxlevfour,\ylevone) {$\MAT{\left\lfloor\frac{t}{2}\right\rfloor + 1}$};

\node(mkcskt) at (\xmaxlevone,\ylevone) {$\MCS{k}{kt}$};

\draw[biorarrow] (spt) -- (ist);
\draw[biorarrow] (ist) -- (asgt);
\draw[strictorarrow] (ist) to[out = 125, in = -125] (is);
\draw[strictorarrow] (scht) -- (sch);
\draw[orarrow] (sch) -- (is);
\draw[biorarrow] (is) -- (cli);
\draw[biorarrow] (scht) -- (spt);
\draw[biorarrow] (scht) -- (clit);
\draw[strictorarrow] (spt) to[out = 55, in = -55] (sp);
\draw[orarrow] (sp) -- (is);
\draw[orarrow] (matt) -- (spt);
\draw[biorarrow] (scht) -- (ist);
\draw[biorarrow] (clit) -- (ist);
\draw[biorarrow] (clit) -- (spt);
\draw[strictorarrow] (clit) to[out = 55, in = -55] (cli);
\draw[orarrow] (mkcskt) -- (ist);

\end{tikzpicture}
\caption[An representation of a hardness graph containing both minimization and maximization problems]{
An (incomplete) representation of the hardness graph for $t \in \ZZ^+$ with respect to $\PAIRMU$. 
The vertices are online problems, and there is an arc $P \rightarrow Q$ when $Q$ is as hard as $P$ (see Definition~\ref{def:as_hard_as}).
If the arrowhead of an arc $P \rightarrow Q$ is only outlined (e.g.\ $\BDIS{t} \rightarrow \IS$) then $Q \rightarrow P$ cannot exist.
All problems and arcs in \textcolor{gray!50}{gray} are from~\cite{BBFL25}.
The remaining problems are defined in Definitions~\ref{def:bdis_t},~\ref{def:sp_t},~\ref{def:sch_t},~\ref{def:cli_t},~\ref{def:mcs_kt}, and~\ref{def:mat_t}, and the existence of the remaining the arcs are proven in Sections~\ref{sec:independent_set_vs_asg_t} and~\ref{sec:a_collection_of_problems_from_CCWM_max}.
All arcs in the transitive closure of the hardness graph exists but are omitted for simplicity.
Hence, if there is a $(P,Q)$-path in the hardness graph, then $P \rightarrow Q$ also exists (see Lemma~\ref{lem:transitivity}).
}
\label{fig:hardness_graph}
\end{figure}


We extend the \emph{hardness graph} from~\cite{BBFL25} to illustrate our results (see Figure~\ref{fig:hardness_graph}).

\section{Preliminaries}\label{sec:preliminaries}

We consider online problems, $P$, where algorithms must output a bit as its answer to each request.
Following~\cite{BBFL25}, we think of an instance of $P$ as a triple $I = (x,\hat{x},r)$, where $x,\hat{x} \in \{0,1\}^n$ are bitstrings and $r = \langle r_1,r_2,\ldots,r_n\rangle$ is a sequence of requests for $P$.
We let $\INSTANCES{P}$ be the set of all instances for $P$, $\ALGS{P}$ be the collection of all deterministic algorithms for $P$, and $\OPT$ be a fixed optimal offline algorithm for $P$.
Given an instance, $I = (x,\hat{x},r)$, $x$ must be an encoding of $\OPT$'s solution and $\hat{x}$ is a prediction of $x$.
When an algorithm for $P$, $\ALG$, receives a request, $r_i$, it also receives the prediction, $\hat{x}_i$, that predicts the optimal output to $r_i$, $x_i$, which $\ALG$ may use to aid its decision about $r_i$, denoted by $y_i$. 
The contents of the requests is defined for each problem separately.

For most graph problems, we consider the \emph{vertex-arrival} model, the most standard for online graph problems~\cite{ABM24,HIMT02,BBBKP18}.
In this model, given an input graph $G = (V,E)$, each request, $r_i$, contains a vertex, $v_i \in V$, and a collection of edges $E_i = \{(v_i,v_j) \mid j < i\} \subseteq E$ to previously revealed vertices.
Throughout, we let $d_G(v)$ be the degree of $v$ in $G$ and we let $\MAXDEGREE(G) = \on{max}_{v \in V}\{d_G(v)\}$ and $\MINDEGREE(G) = \on{min}_{v \in V}\{d_G(v)\}$.


We analyze the performance of algorithms using competitive analysis~\cite{ST85,BE98,K16}. 
An online algorithm, $\ALG$, for an online problem, $P$, is called \emph{$c$-competitive} if there exists a constant $\AT \in \RR$, called the \emph{additive constant},
such that for all instances $I \in \INSTANCES{P}$
\begin{align*}
\OPT(I) \leqslant c \cdot \ALG(I) + \AT
\end{align*}
if $P$ is a maximization problem, or 
\begin{align*}
\ALG(I) \leqslant c \cdot \OPT(I) + \AT
\end{align*}
if $P$ is a minimization problem. 
We extend the definition of competitiveness to take into account the quality of the predictions.
Following previous work on online algorithms with binary predictions~\cite{ABEFHLPS23,BBFL25}, 
we measure the error through a pair of error measures $\PAIRETA$, where $\ETA_b$, for $b \in \{0,1\}$, is a function of the incorrectly predicted requests, where the predictions is $b$. 
For any instance, $I$, we require that $0 \leqslant \ETA_b(I) < \infty$, so infinite and negative error is disallowed.

\begin{definition}\label{def:competitiveness}
Let $P$ be an online problem with binary predictions, let $\PAIRETA$ be a pair of error measures, and let $\alpha,\beta,\gamma \colon \INSTANCES{P} \rightarrow [0,\infty)$ be any maps.
If $P$ is a maximization problem, then $\ALG \in \ALGS{P}$ is called \emph{$\ABC$-competitive with respect to $\PAIRETA$} if there exists $\AT \in \RR$ such that for all $I \in \INSTANCES{P}$,
\begin{align}\label{eq:competitiveness_max}
\OPT(I) \leqslant \alpha \cdot \ALG(I) + \beta \cdot \ETAZERO(I) + \gamma \cdot \ETAONE(I) + \AT.
\end{align}
Similarly, if $P$ is a minimization problem, then $\ALG \in \ALGS{P}$ is called \emph{$\ABC$-competitive with respect to $\PAIRETA$} if there exists $\AT \in \RR$ such that for all $I \in \INSTANCES{P}$,
\begin{align}\label{eq:competitiveness_min}
\ALG(I) \leqslant \alpha \cdot \OPT(I) + \beta \cdot \ETAZERO(I) + \gamma \cdot \ETAONE(I) + \AT,
\end{align}
When $\PAIRETA$ is clear from the context, we say that $\ALG$ is \emph{$\ABC$-competitive}.
\end{definition}

Following~\cite{BBFL25}, we suppress $\alpha$'s, $\beta$'s, and $\gamma$'s dependency on $I$.
Moreover, observe that any $\alpha$-competitive algorithm without predictions for $P$ gives rise to an $(\alpha,0,0)$-competitive algorithm with predictions for $P$ with respect to any pair of error measures, and vice versa.
All results in this paper hold with respect to the canonical pair of error measures~\cite{ABEFHLPS23,BBFL25}, $\PAIRMU$, given by
\begin{align}\label{eq:pairmu}
\MUZERO(I) = \sum_{i=1}^n x_i \cdot (1-\hat{x}_i) \hspace{0.5cm} \text{and} \hspace{0.5cm} \MUONE(I) = \sum_{i=1}^n (1-x_i) \cdot \hat{x}_i,
\end{align}
and the pair of error measures $\PAIRZM$ given by $\ZMZERO(I) = \ZMONE(I) = 0$.
In words, for $b \in \{0,1\}$, $\MU_b$ counts the number of incorrect predictions in $I$, where the prediction is $b$, and $Z_b$ is the zero measure.
The reason that $\PAIRZM$ is interesting is that it models the purely online case, in the sense that an $\ABC$-competitive algorithm with respect to $\PAIRZM$ is $\alpha$-competitive as a purely online algorithm.
Some results also hold with respect to a wider range of error measures (see Theorem~\ref{thm:section_3}).

\subsection[Measuring Hardness and the Complexity Classes $\CCWM{\PAIRETAFORCC}{t}$]{Measuring Hardness and the Complexity Classes $\boldsymbol{\CCWM{\PAIRETAFORCC}{t}}$}\label{sec:complexity_recap}

We recall the central components of the complexity framework from~\cite{BBFL25}.
The aim is to classify online problems based on the relative hardness of online problems, defined as follows:

\begin{definition}[\cite{BBFL25}]\label{def:as_hard_as}
Let $P$ and $Q$ be online problems with binary predictions and error measures $\PAIRETA$ and $\PAIRPHI$.
We say that $Q$ is \emph{as hard as $P$ with respect to $\PAIRPHI$ and $\PAIRETA$}, if the existence of an $\ABC$-competitive algorithm for $Q$ with respect to $\PAIRPHI$ implies the existence of an $\ABC$-competitive algorithm for $P$ with respect to $\PAIRETA$. 
If the error measures are clear from the context, we simply say that $Q$ is \emph{as hard as} $P$.
\end{definition}

In this work, we extend the nomenclature by saying that $Q$ is \emph{exactly as hard as} $P$ if $Q$ is as hard as $P$ and $P$ is as hard as $Q$.

\vspace{0.2cm}
\begin{lemma}[\cite{BBFL25}]\label{lem:transitivity}
The as hard as relation is reflexive and transitive.
\end{lemma}

A main ingredient in the complexity theory from~\cite{BBFL25} is the notion of \emph{strict online reductions}.
For any two online problems $P$ and $Q$, a strict online reduction from $P$ to $Q$ is comprised of two maps; one mapping algorithms for $Q$ to algorithms of $P$, and one mapping instances of $P$ into instances of $Q$.
The mapping of algorithms should be competitiveness preserving, in the sense that an $\ABC$-competitive algorithm for $Q$ maps to an $\ABC$-competitive algorithm for $P$.
In most cases, the algorithm for $P$, say $\ALG$, uses the algorithm for $Q$, say $\ALG'$, as a subroutine, meaning that $\ALG$ gives requests to $\ALG'$ the answers to which $\ALG$ uses to produce its own solution.
The mapping instances of $P$ to instances of $Q$ controls which requests $\ALG$ has given to $\ALG'$, and is mainly used for the analysis.
Observe that the existence of a reduction from $P$ to $Q$ implies that $Q$ is as hard as $P$.
In~\cite{BBFL25}, the authors 
introduce a \emph{reduction template} that provides a method for creating a special type of strict online reductions.

For any $t \in \ZZ^+ \cup \{\infty\}$, the basis of the complexity classes, $\CCWM{\PAIRETAFORCC}{t}$, from~\cite{BBFL25} is \emph{$(1,t)$-Asymmetric String Guessing with Unknown History and Predictions} ($\ASG{t}$).
Given a bitstring, $x$, the task of an algorithm for $\ASG{t}$ is to guess the contents of $x$.
A request for $\ASG{t}$, $r_i$, is a prompt for guessing the next bit, $x_i$.
Together with $r_i$, algorithms for $\ASG{t}$ receive a prediction, $\hat{x}_i$, that predicts $x_i$.
When the algorithm has given a guess, $y_i$, for each bit, $x_i$, in $x$, the full contents of $x$ is revealed to the algorithm.
For any instance $I = (x,\hat{x},r) \in \INSTANCES{\ASG{t}}$, the cost of an algorithm, $\ALG \in \ALGS{\ASG{t}}$, is given by
\begin{align}\label{obj:asg_t}
\ALG(I) = \sum_{i=1}^n \left( y_i + t \cdot x_i \cdot (1 - y_i) \right),
\end{align}
and the goal is to minimize the cost.
Observe that for any $I \in \INSTANCES{\ASG{t}}$, $\OPT(I) = \sum_{i=1}^n x_i$. 

For any $t \in \ZZ^+ \cup \{\infty\}$ and any pair of error measures $\PAIRETA$, the complexity class $\CCWM{\PAIRETAFORCC}{t}$ is defined as the closure of $\ASG{t}$ under the as-hard-as relation. 
Hence, for an online problem, $P$,
\begin{itemize}[label = {-}]
\item $P \in \CCWM{\PAIRETAFORCC}{t}$ if $\ASG{t}$ is as hard as $P$, 
\item $P$ is $\CCWM{\PAIRETAFORCC}{t}$-hard if $P$ is as hard as $\ASG{t}$, and
\item $P$ is $\CCWM{\PAIRETAFORCC}{t}$-complete if $P \in \CCWM{\PAIRETAFORCC}{t}$ and $P$ is $\CCWM{\PAIRETAFORCC}{t}$-hard. 
\end{itemize}

We remark an important detail concerning this complexity setup:




\section{Relative Hardness of Maximization Problems}\label{sec:a_collection_of_problems_from_CCWM_max}

In this section, we consider different online maximization problems, and compare their hardness to the hardness of Online $t$-Bounded Degree Independent Set with Predictions ($\BDIS{t}$).
We start by giving a formal definition of $\BDIS{t}$, then we define our main tool for comparing hardness, called \emph{strict online max-reductions}, and then we prove the existence of a number of strict online max-reductions that compare the hardness of various online maximization problems to the hardness of $\BDIS{t}$.
%
%

\subsection{Independent Set}

Independent Set is one of Karp's original 21 NP-complete problems~\cite{K72}, and it has been studied extensively and in many variants through the years~\cite{HIMT02,DKK12,BDSW24,HR97,BFKM17}.
Given a graph, $G=(V,E)$, the goal of an algorithm for Independent Set is to select a subset of vertices $V_A \subseteq V$ such that for all $u,v \in V_A$, the edge $(u,v)$ is not in $E$.
The goal is to maximize the size of $V_A$.

\begin{definition}\label{def:bdis_t}
\emph{Online $t$-Bounded Degree Independent Set with Predictions ($\BDIS{t}$)} is a vertex-arrival problem.
Given a request, $r_i$, and the prediction $\hat{x}_i$, an algorithm outputs $y_i = 0$ to add the vertex $v_i$ to its independent set, and
$\{v_i \mid x_i = 0\}$ is an optimal independent set. 
For an algorithm's solution to be feasible, it must output an independent set, meaning that
\begin{align*}
y_i = 1 \vee y_j = 1, \mbox{ for each edge } (v_i,v_j) \in E.
\end{align*}
The profit of $\ALG \in \ALGS{\BDIS{t}}$ on instance $I \in \INSTANCES{\BDIS{t}}$ is 
\begin{align}\label{obj:bdis_t}
\ALG(I) = \sum_{i=1}^n (1-y_i),
\end{align}
and the goal is to maximize the profit.
For $t = \infty$, we abbreviate $\BDIS{\infty}$ by $\IS$.
\end{definition}

\subsection{Strict Online Max-Reductions}

We define \emph{strict online max-reductions}, which serves the same purpose as the strict online reductions from~\cite{BBFL25}, except that they are designed for comparing the hardness of online maximization problems instead of online minimization problems.
The purpose of strict online max-reductions is to compare the competitiveness of algorithms for an online problem, $P$, to the competitiveness of algorithms for a different problem, $Q$. 
Intuitively, strict online max-reductions are mappings from, say, $P$ to $Q$, and the existence of one such tells us that $Q$ is as hard as $P$.
To lighten notation, we will usually mark algorithms, instances, and optimal algorithms for the target problem (in this setting $Q$) with a apostrophe, e.g., $\ALG'$, $I'$, and $\OPT'$.

\begin{definition}\label{def:online_max-reduction}
Let $P$ and $Q$ be online maximization problems with predictions, and let $\PAIRETA$ and $\PAIRPHI$ be error measures for the predictions in $P$ and $Q$, respectively.
Let $\rho = (\ORALG{\rho},\ORTRANS{\rho})$ be a tuple consisting of two maps, $\ORALG{\rho}\colon\ALGS{Q}\rightarrow\ALGS{P}$ and $\ORTRANS{\rho}\colon\ALGS{Q} \times \INSTANCES{P} \rightarrow \INSTANCES{Q}$.

If there exists a constant, $a$, called the \emph{reduction term} of $\rho$, such that for each instances $I \in \INSTANCES{P}$ and each algorithm $\ALG' \in \ALGS{Q}$, letting $\ALG = \ORALG{\rho}(\ALG')$ and $I' = \ORTRANS{\rho}(\ALG', I)$, 
%
\begin{enumerate}[label = {(R\arabic*)}]
\item $\ALG'(I') \leqslant \ALG(I)$, \label{item:OR_condition_ALG}
\item $\OPT(I) \leqslant \OPT'(I') + a$, \label{item:OR_condition_OPT_1}
\item $\varphi_0(I') \leqslant \eta_0(I)$, and $\varphi_1(I') \leqslant \eta_1(I)$, \label{item:OR_condition_eta}
\end{enumerate}
then $\rho$ is called a \emph{strict online max-reduction from $P$ to $Q$ with respect to $\PAIRETA$ and $\PAIRPHI$.}
If $\PAIRETA$ and $\PAIRPHI$ are clear from the context, then we simply say that $\rho$ is a \emph{strict online max-reduction from $P$ to $Q$}, and write $\rho \colon P \orarrow Q$.
\end{definition}

For brevity, we often refer to strict online max-reductions simply as max-reductions.
We show that max-reductions can be used as a tool to prove as hard as relations.

\vspace{0.2cm}
\begin{lemma}\label{lem:why_the_naming}
Let $\rho \colon P \orarrow Q$ be a max-reduction as in Definition~\ref{def:online_max-reduction}, let $\ALG' \in \ALGS{Q}$ be an $\ABC$-competitive algorithm for $Q$ with respect to $\PAIRPHI$, and let $\ALG = \ORALG{\rho}(\ALG')$.
Then, $\ALG$ is an $\ABC$-competitive algorithm for $P$ with respect to $\PAIRETA$.
\end{lemma}
\begin{proof}
Since $\ALG'$ is $(\alpha,\beta,\gamma)$-competitive for $Q$ with respect to $\PAIRETA$, there exists $\AT \in \RR$ such that for any $I \in \INSTANCES{P}$, with $I' = \ORTRANS{\rho}(\ALG',I)$, we have that
\begin{align*}
\OPT(I) \overset{\ref{item:OR_condition_OPT_1}}{\leqslant} &\OPT'(I') + a \\
\overset{\textcolor{white}{(R0)} }{\leqslant} &\alpha \cdot \ALG'(I') + \beta \cdot \PHIZERO(I') + \gamma \cdot \PHIONE(I') + \AT + a \\
\overset{\ref{item:OR_condition_ALG}}{\leqslant} &\alpha \cdot \ALG(I) + \beta \cdot \PHIZERO(I') + \gamma \cdot \PHIONE(I') + \AT + a  \\
\overset{\ref{item:OR_condition_eta}}{\leqslant} &\alpha \cdot \ALG(I) + \beta \cdot \ETAZERO(I) + \gamma \cdot \ETAONE(I) + \AT + a,
\end{align*}
where the second inequality uses the $\ABC$-competitiveness of $\ALG'$ (see Definition~\ref{def:competitiveness}).
Hence, $\ALG$ is an $\ABC$-competitive algorithm for $P$ with respect to $\PAIRETA$. 
\end{proof}

To get some intuition for this complexity framework, we start by giving a relatively straight forward reduction, which proves that $\BDIS{t}$ is as hard as an online variant of Set Packing.
For those interested, there are several instructional strict online reductions in~\cite{BBFL25}. 

For the remainder of this section, we prove the existence of several max-reductions.
The full details of two proofs cannot be given before Section~\ref{sec:missing_proofs}, but, for completeness, we have included the statements of the results in this section.
The algorithmic and complexity theoretical consequences of the existence of these max-reductions will be discussed in Section~\ref{sec:consequences}.

\subsection{Set Packing} 

Given a collection of finite sets, $\mathcal{S}$, an algorithm for Set Packing must select a subset $\mathcal{S}_\textsc{a} \subseteq \mathcal{S}$ such that all sets in $\mathcal{S}_\textsc{a}$ are mutually disjoint.
The profit of the solution is the size of $\mathcal{S}_\textsc{a}$ and the goal is to maximize the profit. 

Similarly to Independent Set, Set Packing is one of Karp's original 21 NP-complete problems~\cite{K72}.
Moreover, Set Packing has previously been considered in different variations both in the context of approximation algorithms~\cite{H00,S02} and online algorithms~\cite{EHMPRR12}.

We consider an online variant of Set Packing with predictions:

\begin{definition}\label{def:sp_t}
A request, $r_i$, for \emph{Online $t$-Bounded Set-Arrival Set Packing with Predictions ($\SP{t}$)} is a finite set $S_i$.
Instances for $\SP{t}$ satisfy that any requested set intersects at most $t$ other requested sets.
An algorithm, $\ALG \in \ALGS{\SP{t}}$, outputs $y_i = 0$ to include $S_i$ into its solution, and $\{S_i \mid x_i = 0\}$ is an optimal solution.
For an algorithm's solution to be feasible, all accepted sets must be disjoint, i.e.,
\begin{align*}
y_i = 1 \vee y_j = 1, \mbox{ for each pair } (S_i,S_j) \mbox{ of distinct subsets such that } S_i \cap S_j \neq \emptyset.
\end{align*}
Given an instance $I \in \INSTANCES{\SP{t}}$, the profit of $\ALG$ is
\begin{align*} 
\ALG(I) = \sum_{i=1}^n (1-y_i).
\end{align*}
We abbreviate $\SP{\infty}$ by $\SETPACKING$. 
\end{definition}


The following max-reduction uses a known connection between Set Packing and Independent Set~\cite{K72,GJ90}.
In brief, given a collection of finite sets $\mathcal{S}$ for Set Packing, we create a vertex, $v_S$, for each $S \in \mathcal{S}$, and we include an edge between $v_{S_1}$ and $v_{S_2}$ if $S_1 \cap S_2 \neq \emptyset$. 
Then, a maximum independent set in the graph corresponds to a maximum collection of mutually disjoint sets from $\mathcal{S}$. 

%
%
%
%

\vspace{0.2cm}
\begin{lemma}\label{lem:sp_t_membership}
$\BDIS{t}$ is as hard as $\SP{t}$ for all $t \in \ZZ^+ \cup \{\infty\}$ with respect to all pairs of error measures $\PAIRETA$. 
\end{lemma}
\begin{proof}
We prove the existence of a max-reduction $\rho \colon \SP{t} \orarrow \BDIS{t}$.
To provide a consistent and structured definition and analysis of the max-reductions presented in this paper, we split relevant proof into two parts:
\begin{enumerate}[label = {(\roman*)}]
\item Defining $\ORALG{\rho}$ and $\ORTRANS{\rho}$. \label{item:sp_t_membership_definition_of_the_reduction}
\item Verifying that $\rho = \ONLRED{\rho}$ is a max-reduction. \label{item:sp_t_membership_analysis}
\end{enumerate}

\textbf{Towards~\ref{item:sp_t_membership_definition_of_the_reduction}: Defining $\boldsymbol{\ORALG{\rho}}$ and $\boldsymbol{\ORTRANS{\rho}}$}.
Let $\ALG' \in \ALGS{\BDIS{t}}$ be any algorithm for $\BDIS{t}$ and let $I = (x,\hat{x},r) \in \INSTANCES{\SP{t}}$ be any instance for $\SP{t}$. 
We define $\ALG = \ORALG{\rho}(\ALG')$ and $I' = (x',\hat{x}',r') = \ORTRANS{\rho}(\ALG',I)$ as follows.

When $\ALG$ receives a request containing a set $S_i$ with true and predicted bits $x_i$ and $\hat{x}_i$, respectively, it gives a request to $\ALG'$ containing a vertex $v_i'$ and the edges $E_i' = \{(v_j',v_i') \mid \mbox{$j < i$ and $S_i \cap S_j \neq \emptyset$}\}$ with true and predicted bits $x_i' = x_i$ and $\hat{x}_i' = \hat{x}_i$, respectively. 
If $\ALG'$ accepts $v_i'$, then $\ALG$ accepts $S_i$. 

\textbf{Towards~\ref{item:sp_t_membership_analysis}: Verifying that $\boldsymbol{\rho = \ONLRED{\rho}}$ is a max-reduction.}
Let $G' = (V',E')$ be the underlying graph of $I'$.
Since any set from $I$ intersects at most $t$ other sets from $I$, then $\MAXDEGREE(G') \leqslant t$.
Moreover, by construction, $\mathcal{S}_\textsc{a} = \{S_{i_1},S_{i_2},\ldots,S_{i_k}\}$ is a set of mutually disjoint subsets from $I$ if and only if $V = \{v_{i_1}',v_{i_2}',\ldots,v_{i_k}'\} \subseteq V'$ is an independent set in $G'$.
%
%
%
Hence, $I'$ is a valid instance since $x' = x$ encodes an optimal independent set of $G'$. 

Hence, $\OPT'(I') = \OPT(I)$, for all $I \in \INSTANCES{\SP{t}}$, such that~\ref{item:OR_condition_OPT_1} from Definition~\ref{def:online_max-reduction} is satisfied with $a=0$.
Moreover, since $x' = x$ and $\hat{x}' = \hat{x}$,~\ref{item:OR_condition_eta} is also satisfied.
Finally, by construction, $\ALG'(I') = \ALG(I)$, for all $I \in \INSTANCES{\SP{t}}$,
implying that~\ref{item:OR_condition_ALG} is also satisfied, and so $\rho$ is a max-reduction. 
\end{proof}

The connection between $\BDIS{t}$ and $\SP{t}$ used in the above proof
%
cannot be used directly to create a max-reduction from $\BDIS{t}$ to $\SP{t}$.
In parituclar, when a vertex $v$ arrives, we cannot know which future vertices will be incident to $v$, and therefore we cannot create the set corresponding to $v$ online.
We do, however, use a strategy that is inspired by this idea to create a max-reduction from $\BDIS{t}$ to $\SP{t}$ for all $t \in \ZZ^+$. 

\vspace{0.2cm}
\begin{lemma}\label{lem:sp_t_hardness}
$\SP{t}$ is as hard as $\BDIS{t}$ for all $t \in \ZZ^+$ with respect to all pairs of error measures $\PAIRETA$. 
\end{lemma}
\begin{proof}
We prove the existence of a max-reduction $\rho \colon \BDIS{t} \orarrow \SP{t}$. 
Similarly to the proof of Lemma~\ref{lem:sp_t_membership}, we split the proof into two parts:
\begin{enumerate}[label = {(\roman*)}]
\item Defining $\ORALG{\rho}$ and $\ORTRANS{\rho}$. \label{item:sp_t_hardness_definition_of_the_reduction}
\item Verifying that $\rho = \ONLRED{\rho}$ is a max-reduction. \label{item:sp_t_hardness_analysis}
\end{enumerate}

\textbf{Towards~\ref{item:sp_t_hardness_definition_of_the_reduction}: Defining $\boldsymbol{\ORALG{\rho}}$ and $\boldsymbol{\ORTRANS{\rho}}$.}
Let $\ALG' \in \ALGS{\SP{t}}$ and let $I = (x,\hat{x},r) \in \INSTANCES{\BDIS{t}}$. 
We define $\ALG = \ORALG{\rho}(\ALG')$ and $I' = (x',\hat{x}',r') = \ORTRANS{\rho}(\ALG',I)$ as follows. 

Let $G = (V,E)$ be the underlying graph of $I$.
Recall that a request, $r_i$, for $\BDIS{t}$ contains a vertex, $v_i$, and a collection of edges to previously revealed vertices, $E_i$.
For the rest of this proof, we let $G_i = (\{v_1,v_2,\ldots,v_i\},E_1 \cup E_2 \cup \cdots \cup E_i)$ be the part of $G$ that is known to $\ALG'$ after the first $i$ requests have been given. 
Finally, for $j < i$, we abbreviate $d_{G_i}(v_j)$ by $d_i(v_j)$.

When $\ALG$ receives a request containing the vertex $v_i$ and the set of edges $E_i = \{(v_{j_1},v_i),(v_{j_2},v_i),\ldots,(v_{j_l},v_i)\}$, for some $l \leqslant t$, with true and predicted bits $x_i$ and $\hat{x}_i$, respectively, it gives a request for $\ALG'$ containing the set
\begin{align*}
S_i = \{F_i^{l+1},F_i^{l + 2},\ldots,F_i^t\} \cup \left( \bigcup_{k=1}^{l} \left\{F_{j_k}^{d_i(v_{j_k})}\right\} \right),
\end{align*}
with true and predicted bit $x_i' = x_i$ and $\hat{x}_i' = \hat{x}_i$, respectively.
Then, $\ALG$ accepts $v_i$ if and only if $\ALG'$ accepts $S_i$. 
We give an example of this reduction in Figure~\ref{fig:example_reduction_from_bdist_to_spt}.

\begin{figure}
\centering
\begin{tikzpicture}

\node[vertex, xshift = 3cm] (v1) at (0,0) {$v_1$};
\node[vertex, xshift = 3cm] (v2) at (0,-1.5) {$v_2$};
\node[vertex, xshift = 3cm] (v3) at (3,-1.5) {$v_3$};
\node[vertex, xshift = 3cm] (v4) at (-1.5,-0.75) {$v_4$};
\node[vertex, xshift = 3cm] (v5) at (1.5,0.75) {$v_5$};
\node[vertex, xshift = 3cm] (v6) at (3,0) {$v_6$};

\draw (v1) -- (v2);
\draw (v1) -- (v4);
\draw (v1) -- (v5);
\draw (v1) -- (v6);
\draw (v2) -- (v3);
\draw (v2) -- (v4);
\draw (v2) -- (v6);
\draw (v3) -- (v6);

\node[yshift = -4cm, xshift = -1.25cm, anchor = west] (r1) at (0,1) {$r_1 = (v_1,E_1 = \emptyset)$};
\node[yshift = -4cm, xshift = -1.25cm, anchor = west] (r2) at (0,0.5) {$r_2 = (v_2,E_2 = \{(v_1,v_2)\})$};
\node[yshift = -4cm, xshift = -1.25cm, anchor = west] (r3) at (0,0) {$r_3 = (v_3,E_3 = \{(v_2,v_3)\})$};
\node[yshift = -4cm, xshift = -1.25cm, anchor = west] (r4) at (0,-0.5) {$r_4 = (v_4,E_4 = \{(v_1,v_4),(v_2,v_4)\})$};
\node[yshift = -4cm, xshift = -1.25cm, anchor = west] (r5) at (0,-1) {$r_5 = (v_5,E_5 = \{(v_1,v_5)\})$};
\node[yshift = -4cm, xshift = -1.25cm, anchor = west] (r6) at (0,-1.5) {$r_6 = (v_6,E_6 = \{(v_1,v_6),(v_2,v_6),(v_3,v_6)\})$};

\node[yshift = -4cm, xshift = 5cm, anchor = west] (s1) at (0,1) {$S_1 = \{F_1^1,F_1^2,F_1^3,F_1^4\}$};
\node[yshift = -4cm, xshift = 5cm, anchor = west] (s2) at (0,0.5) {$S_2 = \{F_2^2,F_2^3,F_2^4,F_1^1\}$};
\node[yshift = -4cm, xshift = 5cm, anchor = west] (s3) at (0,0) {$S_3 = \{F_3^2,F_3^3,F_3^4,F_2^2\}$};
\node[yshift = -4cm, xshift = 5cm, anchor = west] (s4) at (0,-0.5) {$S_4 = \{F_4^3,F_4^4,F_1^2,F_2^3\}$};
\node[yshift = -4cm, xshift = 5cm, anchor = west] (s5) at (0,-1) {$S_5 = \{F_5^2,F_5^3,F_5^4,F_1^3\}$};
\node[yshift = -4cm, xshift = 5cm, anchor = west] (s6) at (0,-1.5) {$S_6 = \{F_6^4,F_1^4,F_2^4,F_3^2\}$};

\draw[->] (r1) -- (s1);
\draw[->] (r2) -- (s2);
\draw[->] (r3) -- (s3);
\draw[->] (r4) -- (s4);
\draw[->] (r5) -- (s5);
\draw[->] (r6) -- (s6);

\end{tikzpicture}
\caption{An example of the reduction from $\BDIS{t}$ to $\SP{t}$ with $t = 4$ from Lemma~\ref{lem:sp_t_hardness}.
On the top we have the underlying graph of a $\BDIS{t}$ instance, and on the bottom we translate the requests of the $\BDIS{t}$ instance into the requested sets for the corresponding $\SP{t}$ instance.}
\label{fig:example_reduction_from_bdist_to_spt}
\end{figure}

\textbf{Towards~\ref{item:sp_t_hardness_analysis}: Verifying that $\boldsymbol{\rho = \ONLRED{\rho}}$ is a max-reduction.}
Before proving that~\ref{item:OR_condition_ALG}--\ref{item:OR_condition_eta} from Definition~\ref{def:online_max-reduction} are satisfied, we give some intuition behind the definition of $S_i$. 
By construction, $S_i$ is the union of the two sets:
\begin{align*}
S_i^{\textit{future}} = \{F_i^{l+1},F_i^{l + 2},\ldots,F_i^t\} \hspace{1cm} \text{and} \hspace{1cm} S_i^{\textit{past}} = \bigcup_{k=1}^l \left\{F_{j_k}^{d_i(v_{j_k})}\right\}.
\end{align*}
The first set, $S_i^{\textit{future}}$, contains a collection of flags for future neighbours of $v_i$. 
Observe that the number of flags for $v_i$ is exactly $t - l$; the number of possible future neighbours.
The second set, $S_i^{\textit{past}}$, contains encodings of the edges in $E_i$.
In particular, when the edge $(v_{j_k},v_i) \in E_i$ is revealed, for some $j_k < i$, we want to ensure that $S_i \cap S_{j_k} \neq \emptyset$. 
To this end, we compute the degree of $v_{j_k}$ in $G_i$, $d_i(v_{j_k})$, and then we include the flag $F_{j_k}^{d_i(v_{j_k})}$ into $S_i^{\textit{past}}$.
This represents that $v_i$ is the $d_i(v_{j_k})$'th neighbour of $v_{j_k}$ in $G$.
Observe that the flag $F_{j_k}^{d_i(v_{j_k})}$ exists and is included in $S_{j_k}^{\textit{future}}$ since the degree of $v_{j_k}$ in $G_{j_k}$ is smaller than the degree of $v_{j_k}$ in $G_i$ due to $(v_{j_k},v_i)$ being revealed together with $v_i$.

The reason we cannot simplify this and only have one flag for each vertex, say $S_{j_k}^{\textit{future}} = \{F_{j_k}\}$, and then add $F_{j_k}$ to $S_i^{\textit{past}}$ when the edge $(v_{j_k},v_i)$ is revealed is as follows.
Suppose that two non-adjacent vertices, say $v_i$ and $v_p$, are both adjacent to $v_{j_k}$, with $j_k < i < p$.
Then $F_{i_j}$ would be contained in both $S_i^{\textit{past}}$ and $S_p^{\textit{past}}$.
In this way, $S_i \cap S_p \neq \emptyset$, even though $v_i$ and $v_p$ are non-adjacent in $G$, and possibly both contained in an optimal independent set.
This is not desirable since we cannot directly compare feasible solutions for $\SP{t}$ to feasible solutions for $\BDIS{t}$.
Therefore, we include a number of flags that correspond to the number of possible future neighbours to $S_i^{\textit{future}}$.
In this way, given two non-adjacent vertices that are both adjacent to $v_{j_k}$, $v_i$ and $v_p$, we place different flags for $v_{j_k}$ in $S_i^{\textit{past}}$ and $S_p^{\textit{past}}$, since the degree of $v_{j_k}$ increases as more vertices become adjacent to it.
In this way, we ensure that $S_i \cap S_p = \emptyset$, while $S_i \cap S_{j_k} \neq \emptyset$ and $S_p \cap S_{j_k} \neq \emptyset$. 

Next, we verify that $I'$ is a valid instance of $\SP{t}$.
To this end, we must check two things.
\begin{enumerate}[label = {(\alph*)}]
\item No set from $\mathcal{S}$ intersects more than $t$ other sets, and \label{item:no_set_intersects_more_than_t_other_sets}
\item $x' = x$ encodes an optimal solution to $I'$. \label{item:x_is_optimal_bdist_to_spt}
\end{enumerate}

\textbf{Towards~\ref{item:no_set_intersects_more_than_t_other_sets}:}
By construction, any item from any set from $\mathcal{S}$ is a flag from $\bigcup_{i=1}^n S_i^{\textit{future}}$.
We claim that for any $k \in \{1,2,\ldots,n\}$ and any $p \in \{1,2,\ldots,t\}$, the flag $F_k^p$ (if it exists) is contained in at most two sets in $\mathcal{S}$. 
Indeed, if $F_k^p$ exists, then $F_k^p \in S_k^{\textit{future}}$ by definition. 
Further $F_k^p$ is contained in $S_{r_1}^{\textit{past}}$, for some $r_1 > k$, if and only if $(v_k,v_{r_1}) \in E_r$ and $d_{r_1}(v_k) = p$.
Also, observe that $F_k^p$ cannot be contained in both $S_{r_1}^{\textit{past}}$ and $S_{r_2}^{\textit{past}}$, for some $k < r_1 < r_2$.
Indeed, if $F_k^p \in S_{r_1}^{\textit{past}}$ and $(v_k,v_{r_2}) \in E_{r_2}$, then $d_{r_2}(v_k) > d_{r_1}(v_k)$, since at least one more edge, namely $(v_k,v_{r_2})$, is connected to $v_k$ in $G_{r_2}$.
Hence, the flag for $v_k$ contained in $S_{r_2}^{\textit{past}}$ is marked with a larger degree of $v_k$ than $p$. 

To see that any $S_i \in \mathcal{S}$ only intersect $t$ other sets from $\mathcal{S}$, we notice that $\abs{S_i} = t$, for all $S_i \in \mathcal{S}$, and since any $F_k^l \in S_i$ is contained in at most one other set from $\mathcal{S}$, then $S_i$ cannot intersect more than $t$ other sets from $\mathcal{S}$. 

\textbf{Towards~\ref{item:x_is_optimal_bdist_to_spt}:}
We show this by arguing that $V_A = \{v_{i_1},v_{i_2},\ldots,v_{i_k}\} \subseteq V$ is an independent set if and only of $S_A = \{S_{i_1},S_{i_2},\ldots,S_{i_k}\}$ is a collection of mutually disjoint sets.

Observe that $S_A$ is not a collection of mutually disjoint sets if and only if there exists $a,b \in \{i_1,i_2,\ldots,i_k\}$ with $a < b$ such that $S_a \cap S_b \neq \emptyset$.
Now, $S_a \cap S_b \neq \emptyset$ if and only if $S_a \cap S_b = \{F_a^{d_b(v_a)}\}$, which again is true if and only if $(v_a,v_b) \in E$.
However, $(v_a,v_b) \in E$ if and only if $V_A$ is not and independent set.

The above does not only show that $I' \in \INSTANCES{\SP{t}}$, it also shows that $\OPT(I) = \OPT'(I')$, which implies that Condition~\ref{item:OR_condition_OPT_1} is satisfied with $a=0$.
Also, since $x = x'$ and $\hat{x} = \hat{x}'$, Condition~\ref{item:OR_condition_eta} is satisfied.

Therefore, it only remains to show that~\ref{item:OR_condition_ALG} is satisfied.
To see this, recall that $\ALG$ accepts a vertex $v_i$ if and only if $\ALG'$ accepts $S_i$.
This means that $\ALG(I) = \ALG'(I')$, and so Condition~\ref{item:OR_condition_ALG} is satisfied, which finishes the proof.
\end{proof}


%

\subsection{Clique} 

Given a graph, $G = (V,E)$, an algorithm for Clique must find a subset $V_\textsc{a} \subseteq V$ of vertices such that for any two distinct $v_1,v_2 \in V_\textsc{a}$, $(v_1,v_2) \in E$. 
The profit of the solution is given by the size of $V_\textsc{a}$, and the goal is to maximize the profit.

Similarly to Independent set, Clique is one of Karp's original 21 NP-complete problems~\cite{K72}, and has since been studied in different variations in the context of approximation algorithm~\cite{H99,HR97}, online algorithms~\cite{CDFN20}, and online algorithms with advice~\cite{BFKM17}.

We consider an online variant of Clique with predictions: 

\begin{definition}\label{def:cli_t}
\emph{Online $t$-Bounded Clique with Predictions ($\CLI{t}$)} is a vertex-arrival problem, where all input graphs $G=(V,E)$ satisfy that $\abs{V} - t \leqslant \MINDEGREE(G)$.
An algorithm, $\ALG \in \ALGS{\CLI{t}}$, outputs $y_i = 0$ to include $v_i$ into its clique, and $\{v_i \mid x_i = 0\}$ is an maximum clique. 
For an algorithm's solution to be feasible, it must output a clique, i.e.,
\begin{align*}
\mbox{for each pair } (v_i,v_j) \subset V_A \mbox{ of distinct vertices, } y_i = 0 \wedge y_j = 0.
\end{align*}
The profit of $\ALG \in \ALGS{\CLI{t}}$ on instance $I \in \INSTANCES{\CLI{t}}$ is
\begin{align}\label{obj:cli_t}
\ALG(I) = \sum_{i=1}^n (1-y_i),
\end{align}
and the goal is to maximize the profit.x
We abbreviate $\CLI{\infty}$ by $\CLIQUE$. 
\end{definition}

There has been previous work studying Clique on graphs with bounded minimum degree~\cite{L12}.
Parametrizing Online Clique by bounding the minimum degree on the input graphs allows us to prove that $\CLI{t}$ is exactly as hard as $\BDIS{t}$ for all $t \in \ZZ^+ \cup \{\infty\}$ and all pairs of error measures $\PAIRETA$, rather than only proving this result for $t = \infty$.

\vspace{0.2cm}
\begin{lemma}\label{lem:clique_completeness}
$\CLI{t}$ is exactly as hard as $\BDIS{t}$ for all $t \in \ZZ^+ \cup \{\infty\}$ with respect to all pairs of error measures $\PAIRETA$.
\end{lemma}
\begin{proof}
We show that $\BDIS{t}$ is as hard as $\CLI{t}$ and that $\CLI{t}$ is as hard as $\BDIS{t}$ by establishing two max-reductions $\rho \colon \CLI{t} \orarrow \BDIS{t}$ and $\tau \colon \BDIS{t} \orarrow \CLI{t}$ that holds with respect to all pairs of error measures.
The structures of $\rho$ and $\tau$ are basically identical, and therefore we only prove the existence of $\rho$.

Following the structure from previous lemmas, we split this proof into two parts:
\begin{enumerate}[label = {(\roman*)}]
\item Defining $\ORALG{\rho}$ and $\ORTRANS{\rho}$. \label{item:clique_completeness_definition_of_the_reduction}
\item Verifying that $\rho = \ONLRED{\rho}$ is a max-reduction. \label{item:clique_completeness_analysis}
\end{enumerate}

\textbf{Towards~\ref{item:clique_completeness_definition_of_the_reduction}: Defining $\boldsymbol{\ORALG{\rho}}$ and $\boldsymbol{\ORTRANS{\rho}}$.}
Let $\ALG' \in \ALGS{\BDIS{t}}$, let $I = (x,\hat{x},r) \in \INSTANCES{\CLI{t}}$, and let $G = (V,E)$ be the underlying graph of $I$.
We define $\ALG = \ORALG{\rho}(\ALG')$ and $I' = (x',\hat{x}',r') = \ORTRANS{\rho}(\ALG',I)$ as follows.

When $\ALG$ receives a request containing the vertex, $v_i$, and a collection of edges, $E_i$, together with a prediction, $\hat{x}_i$, it gives a request to $\ALG'$ containing a new vertex, $v_i'$, and the set of edges $E_i' = \{(v_j',v_i') \mid \mbox{$j < i$ and $(v_j,v_i) \not\in E_i$}\}$ with true and predicted bits $x_i' = x_i$ and $\hat{x}_i' = \hat{x}_i$, respectively. 
Then, $\ALG$ outputs the same for $v_i$ as $\ALG'$ outputs for $v_i'$. 

\textbf{Towards~\ref{item:clique_completeness_analysis}: Verifying that $\boldsymbol{\rho = \ONLRED{\rho}}$ is a max-reduction.}
By definition, the underlying graph of $I'$, denoted $G' = (V',E')$, is the compliment graph of $G$.
Since $\abs{V} - t \leqslant \MINDEGREE(G)$, then $d_{G}(v) \geqslant n - t$ for all $v \in V$, and so $d_{G'}(v') \leqslant n - (n-t) = t$ for all $v' \in V'$.
Hence, $\MAXDEGREE(G') \leqslant t$ which is required to be a valid instance of $\BDIS{t}$.
Further, since a maximum clique in $G$ encodes a maximum independent set in $G'$, then $I' \in \INSTANCES{\BDIS{t}}$, as the sequence of true bits, $x'$, encodes an optimal solution to $G'$.
Hence, $\OPT'(I') = \OPT(I)$ for all $I \in \INSTANCES{\CLI{t}}$, and so, Condition~\ref{item:OR_condition_OPT_1} from Definition~\ref{def:online_max-reduction} is satisfied with $a = 0$.
Further, since $x' = x$ and $\hat{x}' = \hat{x}$, Condition~\ref{item:OR_condition_eta} is also satisfied.
Finally, by definition of $\ALG$, it follows that, $\ALG(I) = \ALG'(I')$ for all $I \in \INSTANCES{\CLI{t}}$, and so Condition~\ref{item:OR_condition_ALG} is satisfied, which proves that $\rho$ is a max-reduction.
\end{proof}

\subsection{Interval Scheduling}\label{sec:online_interval_scheduling}

Given a collection of open intervals on the real line, $\mathcal{I}$, an algorithm for Interval Scheduling must select a subset $\mathcal{I}_\textsc{a} \subseteq \mathcal{I}$ of intervals such that no two intervals in $\mathcal{I}_\textsc{a}$ overlap.
The profit of the solution is the size of $\mathcal{I}_\textsc{a}$ and the goal is to maximize the profit.

Interval Scheduling is a polynomially solvable problem in the offline setting~\cite{KT05}, which turns out to be very hard in the online setting~\cite{LT94}.
Other work on Interval Scheduling includes online algorithms with advice~\cite{GKKKS15} and online algorithms with predictions~\cite{BFKL23,K25}.
Observe that the prediction scheme from~\cite{BFKL23} is much richer than the one considered here, and the prediction scheme from~\cite{K25} is the same as the one considered here. 

We consider an online variant of Interval Scheduling with predictions:

\begin{definition}\label{def:sch_t}
A request, $r_i$, for \emph{Online $t$-Bounded Overlap Interval Scheduling with Predictions ($\SCH{t}$)} is an open interval on the real line. 
Instances for $\SCH{t}$ satisfy that any requested interval, $r_i$, overlap at most $t$ other requested intervals.
An algorithm, $\ALG \in \ALGS{\SCH{t}}$, outputs $y_i = 0$ to accept $r_i$ into its solution and $\{r_i \mid x_i = 0\}$ is an optimal schedule.
For an algorithm's solution to be feasible, it must output a valid schedule, i.e.,
\begin{align*}
y_i = 1 \vee y_j = 1 \mbox{ for all distinct } r_i,r_j \in S_A.
\end{align*}
The profit of $\ALG \in \ALGS{\SCH{t}}$ on instance $I \in \INSTANCES{\SCH{t}}$ is
\begin{align*} 
\ALG(I) = \sum_{i=1}^n (1-y_i),
\end{align*}
and the goal is to maximuze the profit.
For $t = \infty$, we abbreviate $\SCH{t}$ by $\SCHEDULING$. 
\end{definition}


\vspace{0.2cm}
\begin{lemma}\label{lem:sch_t_to_bdis_t}
$\BDIS{t}$ is as hard as $\SCH{t}$ for all $t \in \ZZ^+ \cup \{\infty\}$ with respect to all pairs of error measures $\PAIRETA$.
\end{lemma}
\begin{proof}
We give a max-reduction $\rho \colon \SCH{t} \orarrow \BDIS{t}$. 
For consistency, we split this proof into two parts:
\begin{enumerate}[label = {(\roman*)}]
\item Defining $\ORALG{\rho}$ and $\ORTRANS{\rho}$. \label{item:sch_t_to_bdis_t_definition_of_the_reduction}
\item Verifying that $\rho = \ONLRED{\rho}$ is a max-reduction. \label{item:sch_t_to_bdis_t_completeness_analysis}
\end{enumerate}

\textbf{Towards~\ref{item:clique_completeness_definition_of_the_reduction}: Defining $\boldsymbol{\ORALG{\rho}}$ and $\boldsymbol{\ORTRANS{\rho}}$.}
Let $\ALG' \in \ALGS{\BDIS{t}}$ and let $I = (x,\hat{x},r) \in \INSTANCES{\SCH{t}}$.
We define $\ALG = \ORALG{\rho}(\ALG')$ and  $I' = (x',\hat{x}',r') = \ORTRANS{\rho}(\ALG',I) \in \INSTANCES{\BDIS{t}}$ as follows. 
When $\ALG$ receives an interval $r_i$ with true and predicted bits $x_i$ and $\hat{x}_i$, we give a request to $\ALG'$ for a new vertex, $v_i'$, with the edges $E_i' = \{ (v_j',v_i') \mid \mbox{$j < i$ and $r_j \cap r_i \neq \emptyset$}\}$.
The true and predicted bits of $v_i'$ are $x_i' = x_i$ and $\hat{x}_i' = \hat{x}_i$, respectively.
Then, $\ALG$ accepts $r_i$ if and only if $\ALG'$ accepts $v_i'$.

\textbf{Towards~\ref{item:clique_completeness_analysis}: Verifying that $\boldsymbol{\rho = \ONLRED{\rho}}$ is a max-reduction.}
We start by verifying that $I' \in \INSTANCES{\BDIS{t}}$.
To this end, let $G' = (V',E')$ be the underlying graph of $I'$, and observe that since all intervals in $I$ overlap at most $t$ other intervals we have that $\MAXDEGREE(G') \leqslant t$.
Next, observe that $\mathcal{I}_\textsc{a} = \{r_{i_1},r_{i_2},\ldots,r_{i_l}\}$ is a collection of non-overlapping intervals if and only if $V = \{v_{i_1}',v_{i_2}',\ldots,v_{i_l}'\} \subseteq V'$ is an independent set in $G$.
Hence, $x'$ encodes an optimal solution to $I'$, and so $I' \in \INSTANCES{\BDIS{t}}$.

This, in turn, implies that $\OPT(I) = \OPT'(I')$ for all $I \in \INSTANCES{\SCH{t}}$, which means that Condition~\ref{item:OR_condition_OPT_1} is satisfied with $a = 0$.
Further, since $x' = x$ and $\hat{x}' = \hat{x}$, Condition~\ref{item:OR_condition_eta} is also satisfied.
To see that Condition~\ref{item:OR_condition_ALG} is also satisfied, we observe that $\ALG(I) = \ALG'(I')$ for all $I \in \INSTANCES{\SCH{t}}$ by construction.
\end{proof}

Our next goal is to prove that $\SCH{t}$ is also as hard as $\BDIS{t}$.
Intuitively, this should not be true since $\SCH{t}$ is equivalent Online $t$-Bounded Degree Independent Set with Predictions on interval graphs; a subproblem of $\BDIS{t}$.
However, as we will see in Lemma~\ref{lem:bdis_t_to_sch_t}, in the online setting with predictions, $\SCH{t}$ is just as hard to solve as $\BDIS{t}$ with respect to the canonical pair of error measures $\PAIRMU$.
To prove this result, however, we need a technical lemma that cannot be proven until after Section~\ref{sec:independent_set_vs_asg_t}. 
For completeness, we state the result here and postpone the proof until Section~\ref{sec:proof_of_lemma_bdis_t_to_sch_t}.



\vspace{0.2cm}
\begin{restatable}{lemma}{bdisttoscht}\label{lem:bdis_t_to_sch_t}
$\SCH{t}$ is as hard as $\BDIS{t}$ for all $t \in \ZZ^+$ with respect to $\PAIRMU$.
\end{restatable}


\vspace{0.2cm}
\begin{theorem}\label{thm:sch_t_completeness}
$\SCH{t}$ is exactly as hard as $\BDIS{t}$ for all $t \in \ZZ^+$ with respect to $\PAIRMU$.
\end{theorem}

\subsection[Maximum $k$-Colorable Subgraph]{Maximum $\boldsymbol{k}$-Colorable Subgraph}

Given a graph, $G = (V,E)$, an algorithm for Maximum $k$-Colorable Subgraph must determine a subset $V_\textsc{a} \subseteq V$ of vertices such that the subgraph of $G$ induced by $V_\textsc{a}$ is $k$-colorable.
The profit of the solution is given by the size of $V_\textsc{a}$, and the goal is to maximize the profit.

Previous work has concluded that Maximum $k$-Colorable Subgraph is NP-complete~\cite{GJ90}, and different variations of the problem has since been studied in the context of approximation algorithms~\cite{N89} and online algorithms~\cite{GHM18}. 

We consider an online variant of Maximum $k$-Colorable Subgraph with predictions:

\begin{definition}\label{def:mcs_kt}
\emph{Online $t$-Bounded Maximum $k$-Colorable Subgraph with Predictions ($\MCS{k}{t}$)} is a vertex-arrival problem, where all input graphs $G = (V,E)$ satisfy that $\MAXDEGREE(G) \leqslant t$.
An algorithm, $\ALG \in \ALGS{\MCS{k}{t}}$, outputs $y_i = 0$ to include $v_i$ into its solution, and the subgraph of $G$ induced by $\{v_i \mid x_i = 0\}$ is a maximum $k$-colorable subgraph.
For an algorithm's solution to be feasible, it must induce a $k$-colorable subgraph of $G$.

The profit of $\ALG \in \ALGS{\MCS{k}{t}}$ on instance $I \in \INSTANCES{\MCS{k}{t}}$ is
\begin{align*} 
\ALG(I) = \sum_{i=1}^n (1-y_i),
\end{align*}
and the goal is to maximize the profit.
For $t = \infty$, we abbreviate $\MCS{k}{t}$ by $\MCS{k}{}$.
\end{definition}


One may observe that $\MCS{1}{t}$ is equivalent to $\BDIS{t}$.

For proving our main result on $\MCS{k}{t}$, we use a strategy that is similar to that for proving hardness of $\SCH{t}$. 
Therefore, similarly to Lemma~\ref{lem:bdis_t_to_sch_t}, we state the result here, and postpone the proof until Section~\ref{sec:proof_of_lemma_more_colors_is_easier}.
%
%

\vspace{0.2cm}
\begin{restatable}{lemma}{morecolorsiseasier}\label{lem:more_colors_is_easier}
$\BDIS{t}$ is as hard as $\MCS{k}{kt}$ for all $k,t \in \ZZ^+$ with respect to $\PAIRMU$.
\end{restatable}

\subsection{Maximum Matching} 

Given a graph, $G = (V,E)$, an algorithm for Maximum Matching must determine a subset $E_\textsc{a} \subseteq E$ of edges such that no two edges in $E_\textsc{a}$ share an endpoint.
The profit of a solution is given by the size of $E_\textsc{a}$ and the goal is to maximize the profit.

Maximum Matching is, similarly to Interval Scheduling, polynomially solvable in the offline case~\cite{E65} and has also been studied in different variations in the context of approximation algorithms~\cite{DKPU18} and online algorithms~\cite{GKMSW19,HKTWZZ20}.

We consider an online variant of Maximum Matching with predictions:

\begin{definition}\label{def:mat_t}
A request, $r_i$, for \emph{Online $t$-Bounded Maximum Matching with Predictions ($\MAT{t}$)} is an edge, $e_i$, together with its endpoints.
All input graphs for $\MAT{t}$ satisfy that $\MAXDEGREE(G) \leqslant t$.
An algorithm, $\ALG \in \ALGS{\MAT{t}}$, outputs $y_i = 0$ to include $e_i$ into its solution, and $\{e_i \mid x_i = 0\}$ is a maximum matching in $G$.
For an algorithm's solution to be feasible, it must output a maximum matching, i.e.,
\begin{align*}
y_i = 1 \vee y_j = 0 \mbox{ if $e_i$ and $e_j$ share an endpoint.}
\end{align*}

The profit of $\ALG \in \ALGS{\MAT{t}}$ on instance $I \in \INSTANCES{\MAT{t}}$ is
\begin{align*} 
\ALG(I) = \sum_{i=1}^n (1-y_i),
\end{align*}
and the goal is to maximize the profit.
For $t = \infty$, we abbreviate $\MAT{t}$ by $\MATCHING$. 
\end{definition}

\vspace{0.2cm}
\begin{lemma}\label{lem:membership_matching}
$\SP{t}$ is as hard as $\MAT{\left\lfloor t/2\right\rfloor+1}$ for all $t \in \ZZ^+ \cup \{\infty\}$ with respect to all pairs of error measures $\PAIRETA$.
\end{lemma}
\begin{proof}
For ease of notation, let $\tilde{t} = \left\lfloor \frac{t}{2} \right\rfloor + 1$.
We give a max-reduction $\rho \colon \MAT{\tilde{t}} \orarrow \SP{t}$.
For consistency, we split this proof into two parts:
\begin{enumerate}[label = {(\roman*)}]
\item Defining $\ORALG{\rho}$ and $\ORTRANS{\rho}$. \label{item:sch_t_to_bdis_t_definition_of_the_reduction}
\item Verifying that $\rho = \ONLRED{\rho}$ is a max-reduction. \label{item:sch_t_to_bdis_t_completeness_analysis}
\end{enumerate}

\textbf{Towards~\ref{item:clique_completeness_definition_of_the_reduction}: Defining $\boldsymbol{\ORALG{\rho}}$ and $\boldsymbol{\ORTRANS{\rho}}$.}
Let $\ALG' \in \ALGS{\SP{t}}$ and let $I \in \INSTANCES{\MAT{\tilde{t}}}$. 
We define $\ALG = \ORALG{\rho}(\ALG')$ and $I' = (x',\hat{x}',r') = \ORTRANS{\rho}(\ALG',I)$ as follows. 
When $\ALG$ receives an edge $e_i = (u,v)$ with true and predicted bits $x_i$ and $\hat{x}_i$, respectively, it gives a request to $\ALG'$ containing a new set $S_i = \{u,v\}$ with true and predicted bits $x_i' = x_i$ and $\hat{x}_i' = \hat{x}_i$, respectively. 
Then, $\ALG$ accepts $e_i$ if and only if $\ALG'$ accepts $S_i$.

\textbf{Towards~\ref{item:clique_completeness_analysis}: Verifying that $\boldsymbol{\rho = \ONLRED{\rho}}$ is a max-reduction.}
Let $G = (V,E)$ be the underlying graph of $I$ and let $\mathcal{S}$ be the collection of sets requested in $I'$.
Then, for any $S_i \in \mathcal{S}$, $S_i$ contain the endpoints of $e_i = (u,v)$. 
Since $\MAXDEGREE(G) \leqslant \tilde{t}$, then $d_G(u) \leqslant \tilde{t}$ and $d_G(v) \leqslant \tilde{t}$.
Since $(u,v) \in E$, then $u$ is adjacent to at most $\tilde{t}-1 = \left\lfloor\frac{t}{2}\right\rfloor$ vertices in $V \setminus \{v\}$.
Similarly, $v$ is adjacent to at most $\left\lfloor\frac{t}{2}\right\rfloor$ vertices in $V \setminus \{u\}$.
Hence, $u$ is contained in at most $\left\lfloor\frac{t}{2}\right\rfloor$ other sets than $S_i$ and $v$ is contained in at most $\left\lfloor\frac{t}{2}\right\rfloor$ other sets than $S_i$. 
Therefore $S_i$ intersects at most $t$ other sets in $\mathcal{S}$. 

By construction, it is clear that the set $\{e_{i_1},e_{i_2},\ldots,e_{i_k}\}$ is a matching if, and only if, $\{S_{i_1},S_{i_2},\ldots,S_{i_k}\}$ is a collection of mutually disjoint sets.
Thus, $I' \in \INSTANCES{\SP{t}}$ is a valid instance, and we have that $\OPT(I) = \OPT'(I')$, meaning that Condition~\ref{item:OR_condition_OPT_1} is satisfied with $a=0$.
Moreover, since $x' = x$ and $\hat{x}' = \hat{x}$, Condition~\ref{item:OR_condition_eta} is also satisfied. 
Finally, $\ALG_{\MAT{\tilde{t}}}(I) = \ALG_{\SP{t}}(\ORTRANS{\rho}(\ALG_{\SP{t}},I))$, by definition of $\ALG_{\MAT{\tilde{t}}}$, meaning that Condition~\ref{item:OR_condition_ALG} is satisfied, which concludes the proof.
\end{proof}

\subsection{Summary}

We give a summary of (a selection of) the results we have proven in Section~\ref{sec:a_collection_of_problems_from_CCWM_max}.
Since the as hard as relation is transitive, there are of course many more relations than the ones stated here.
We highlight the results that correspond to the arrows depicted in the hardness graph (Figure~\ref{fig:hardness_graph}).

\vspace{0.2cm}
\begin{theorem}\label{thm:section_3}
For all $t \in \ZZ^+$ and with respect to any pair of error measures, we have that

\begin{minipage}{\textwidth}
\begin{multicols*}{2}
\begin{enumerate}[label = {(\alph*)}]
\item $\BDIS{t}$ is exactly as hard as $\SP{t}$, \label{item:a}
\item $\BDIS{t}$ is exactly as hard as $\CLI{t}$, \label{item:b}
\item $\BDIS{t}$ is as hard as $\SCH{t}$, \label{item:c}
\item $\BDIS{t}$ is as hard as $\MAT{\left\lfloor t/2 \right\rfloor + 1}$, \label{item:d}
\item $\IS$ is as hard as $\BDIS{t}$, \label{item:e}
\item $\SETPACKING$ is as hard as $\SP{t}$, \label{item:f}
\item $\IS$ is as hard as $\SETPACKING$, \label{item:g}
\item $\CLIQUE$ is as hard as $\CLI{t}$, \label{item:h}
\item $\CLIQUE$ is exactly as hard as $\IS$, \label{item:i}
\item $\SCHEDULING$ is as hard as $\SCH{t}$, and \label{item:j}
\item $\IS$ is as hard as $\SCHEDULING$. \label{item:k}
\end{enumerate}
\end{multicols*}
\end{minipage}

Furthermore, for all $t,k \in \ZZ^+$ and with respect to $\PAIRMU$, we have that

\begin{minipage}{\textwidth}
\begin{multicols*}{2}
\begin{enumerate}[label = {(\alph*)}, start = 12]
\item $\BDIS{t}$ is exactly as hard as $\SCH{t}$, \label{item:l}
\item $\BDIS{t}$ is as hard as $\MCS{k}{kt}$ \label{item:m}
\end{enumerate}
\end{multicols*}
\end{minipage}

\end{theorem}
\begin{proof}
Items~\ref{item:a} and~\ref{item:g} follows from Lemmas~\ref{lem:sp_t_membership} and~\ref{lem:sp_t_hardness}.
Items~\ref{item:b} and~\ref{item:i} follows from Lemma~\ref{lem:clique_completeness}.
Items~\ref{item:c} and~\ref{item:k} follows from Lemma~\ref{lem:sch_t_to_bdis_t}.
Item~\ref{item:d} follows from Lemma~\ref{lem:membership_matching}.
Towards Item~\ref{item:e}, observe that $\INSTANCES{\BDIS{t}} \subset \INSTANCES{\IS}$, meaning that $\rho = \ONLRED{\rho}$ given by $\ALG = \ORALG{\rho}(\ALG)$ and $I = \ORTRANS{\rho}(\ALG,I)$, for any $\ALG \in \ALGS{\IS}$ and any $I \in \INSTANCES{\BDIS{t}}$ is a valid max-reduction. Therefore $\IS$ is as hard as $\BDIS{t}$.
Items~\ref{item:f},~\ref{item:h}, and~\ref{item:j} are proven similarly to Item~\ref{item:e}.
Item~\ref{item:l} follows from Lemmas~\ref{lem:sch_t_to_bdis_t} and~\ref{lem:bdis_t_to_sch_t}.
Finally, Item~\ref{item:m} follows from Lemma~\ref{lem:more_colors_is_easier}.
\end{proof}

\section{Comparing the Hardness of $\BDIS{t}$ and $\ASG{t}$}\label{sec:independent_set_vs_asg_t}

In this section, we give our main result which is a comparison between the hardness of $\BDIS{t}$ and the hardness of $\ASG{t}$, the base problem of $\CCWM{\PAIRMUFORCC}{t}$, while respecting that hardness is measured differently for minimization problems and maximization problems (see Definition~\ref{def:competitiveness}).

In this section we compare the profit algorithms for $\BDIS{t}$ to the cost of algorithms for $\ASG{t}$.
To avoid confusion, we denote algorithms for $\BDIS{t}$ by $\ALG_{\BDIS{t}}$ and algorithms for $\ASG{t}$ by $\ALG_{\ASG{t}}$.



\subsection{$\ASG{t}$ is as hard as $\BDIS{t}$}

First, we show that $\ASG{t}$ is as hard as $\BDIS{t}$ with respect to $\PAIRMU$, by building a reduction, $\rho$, from $\BDIS{t}$ to $\ASG{t}$.
For this, we need to recall a definition and a result from~\cite{BBFL25}:

\begin{definition}\label{def:pareto_optimal}
An $\ABC$-competitive algorithm is called \emph{Pareto-optimal} for a problem, $P$, if, for any $\eps > 0$, there cannot exist an $(\alpha-\eps,\beta.\gamma)$-, an $(\alpha,\beta-\eps,\gamma)$-, or an $(\alpha,\beta,\gamma-\eps)$-competitive algorithm for $P$. 
\end{definition}


%
%

\vspace{0.2cm}
\begin{theorem}[{\cite[Theorem 3]{BBFL25}}]\label{thm:pareto_optimal_algorithms_for_asg_t}
Any $\ABC$-competitive Pareto-optimal algorithm for $\ASG{t}$ with respect to $\PAIRMU$ has $\gamma \leqslant 1$.
\end{theorem}

The following reduction from $\BDIS{t}$ to $\ASG{t}$ will be built similarly to a max-reduction, in the sense that we define two maps $\ORALG{\rho}$ and $\ORTRANS{\rho}$, that translate algorithms for $\ASG{t}$ into algorithms for $\BDIS{t}$ and instances of $\BDIS{t}$ into instances of $\ASG{t}$.
However, since we are not reducing from a maximization (reps.\ minimization) problem to a maximization (resp.\ minimization) problem, we cannot simply define a strict online max-reduction (resp.\ strict online reduction) between the two, ass the competitiveness of algorithms for maximization problems and minimization problems are measured differently (see Definition~\ref{def:competitiveness}).
For this reason, we use a different from of analysis compared to those used in Section~\ref{sec:a_collection_of_problems_from_CCWM_max}.

\vspace{0.2cm}
\begin{lemma}\label{lem:reduction_bdis_t_to_asg_t}
$\ASG{t}$ is as hard as $\BDIS{t}$ for any $t \in \ZZ^+$ with respect to $\PAIRMU$.
\end{lemma}
\begin{proof}
Throughout the proof, we only measure prediction error using the canonical error measures $\PAIRMU$ defined in Equation~\eqref{eq:pairmu}. 
Therefore, we suppress notation and nomenclature related to error measures whenever possible.

We define a tuple $\rho = \ONLRED{\rho}$ that closely resembles a max-reduction from $\BDIS{t}$ to $\ASG{t}$, without actually being one.
Therefore, similarly to previous proofs, we split this proof into several parts:
\begin{enumerate}[label = {(\roman*)}]
\item \label{item:define_maps} Define two maps $\ORALG{\rho} \colon \ALGS{\ASG{t}} \rightarrow \ALGS{\BDIS{t}}$ and $\ORTRANS{\rho} \colon \ALGS{\ASG{t}} \times \INSTANCES{\BDIS{t}} \rightarrow \INSTANCES{\ASG{t}}$.
\item \label{item:pareto_optimal_reduction} We show that if $\ALG_{\ASG{t}} \in \ALGS{\ASG{t}}$ is an $\ABC$-competitive Pareto-optimal algorithm for $\ASG{t}$, then $\ALG_{\BDIS{t}} = \ORALG{\rho}(\ALG_{\ASG{t}})$ is an $\ABC$-competitive for $\BDIS{t}$.
We ensure that our reduction respects that when an algorithm is $\ABC$-competitive for $\BDIS{t}$, it satisfies Equation~\eqref{eq:competitiveness_max}, and when an algorithm is $\ABC$-competitive for $\ASG{t}$, it satisfies Equation~\eqref{eq:competitiveness_min}.
\item \label{item:extend_beyond_pareto_optimality} We show that $\ASG{t}$ is as hard as $\BDIS{t}$.
For this, we show that the existence of a (possibly non-Pareto-optimal) $\ABC$-competitive algorithm for $\ASG{t}$, implies the existence of an $\ABC$-competitive algorithm for $\BDIS{t}$ as well.
\end{enumerate}

%

\textbf{Towards~\ref{item:define_maps}: Defining $\boldsymbol{{\ORALG{\rho}}}$ and $\boldsymbol{{\ORTRANS{\rho}}}$.}
Let $\ALG_{\ASG{t}} \in \ALGS{\ASG{t}}$, let $I = (x,\hat{x},r) \in \INSTANCES{\BDIS{t}}$, and let $G$ be the underlying graph of $I$.
Then, we define $\ALG_{\BDIS{t}} = \ORALG{\rho}(\ALG_{\ASG{t}})$ and $I' = (x',\hat{x}',r') = \ORTRANS{\rho}(\ALG_{\ASG{t}},I)$ as follows.


Whenever $\ALG_{\BDIS{t}}$ receives a vertex $v_i$ with prediction $\hat{x}_i$ as part of processing $I$, it determines the \emph{level of $v_i$}, denoted $\ell(v_i)$.
If there exists a vertex $v_j$, for some $j < i$, with $y_j = 0$ and $(v_i,v_j) \in E_i$, then $\ALG_{\BDIS{t}}$ places $v_i$ on level $2$, and rejects $v_i$ (i.e.\ outputs $y_i = 1$).
Otherwise, $\ALG_{\BDIS{t}}$ places $v_i$ on level $1$, and asks $\ALG_{\ASG{t}}$ to guess the next the next bit given $\hat{x}_i$.
Then, $\ALG_{\BDIS{t}}$ outputs the same bit as $\ALG_{\ASG{t}}$.
We give an example of this reduction in Figure~\ref{fig:bdis_t_to_asg_t}. 

\begin{figure}
\centering
\begin{tikzpicture}

\def\sx{7.2} 
\def\sy{1.5} 
\def\sw{4} 
\def\sh{3} 
\def\offset{0.2}

\def\dist{3}

\def\bbx{\sx + \sw + \dist} 
\def\bby{1.75} 
\def\bbz{0} 
\def\bbsl{1.2} 

\def\vdistx{1.5}
\def\vdisty{1}

\draw[thick, fill = black!50] (\sx + \sw*0.5 - \sw*0.1, \sy-\offset) -- (\sx + \sw*0.5 - \sw*0.1,\sy - \offset - 0.3*\sh) -- (\sx + \sw*0.5 + \sw*0.1,\sy - \offset - 0.3*\sh) --(\sx + \sw*0.5 + \sw*0.1,\sy - \offset) ;
\draw[thick, fill = black!50] (\sx + \sw*0.25,\sy - \offset - 0.3*\sh) rectangle (\sx + \sw*0.5 + \sw*0.25,\sy - \offset - 0.3*\sh - 0.2);
\draw[thick, rounded corners,fill = black!75] (\sx-\offset,\sy-\offset) rectangle (\sx + \sw + \offset,\sy + \sh + \offset);
\draw[thick,fill = white] (\sx,\sy) rectangle (\sx + \sw,\sy + \sh);
\node[prompt] at (\sx + 0.1, \sy + \sh - 0.25) {$>$};

\node at (\sx + 0.5*\sw,\sy + \sh + 0.5) {$\ALG_{\ASG{t}}$};

\node at (2.5,6.25) {$\ALG_{\BDIS{t}}$};
\draw[rounded corners] (0,0) -- (0,6) -- (5,6) -- (5,0) -- cycle;
\draw[dashed] (0.2,3) -- (4.8,3);
\node[anchor = west] at (0,5.75) {Level $2$};
\node[anchor = west] at (0,0.25) {Level $1$};

\draw[dashed, ->] (5 + 0.5*\offset,3.5) -- (\sx - 1.5*\offset , 3.5);
\draw[dashed, ->] (\sx - 1.5*\offset , 2.5) -- (5 + 0.5*\offset,2.5);

\node[vertex,very thick] (v1) at (0.75,2) {$v_1$};
\node at (0.75,1.4) {$\hat{x}_1 = 0$};
\node[vertex] (v2) at (1.75,4) {$v_2$};
\node at (1.75,4.6) {$\hat{x}_2 = 1$};
\node[vertex] (v3) at (2.75,2) {$v_3$};
\node at (2.75,1.4) {$\hat{x}_3 = 1$};
\node[vertex] (v4) at (4.25,2) {$v_4$};
\node at (4.25,1.4) {$\hat{x}_4 = 1$};

\node[prompt] at (\sx + 0.1, \sy + \sh - 0.25) {$> \hat{x}_1 = 0$; Guess the next bit: {$0$}};
\node[prompt] at (\sx + 0.1, \sy + \sh - 0.75) {$> \hat{x}_3 = 1$; Guess the next bit: {$1$}};
\node[prompt] at (\sx + 0.1, \sy + \sh - 1.25) {$> \hat{x}_4 = 1$; Guess the next bit: {$1$}};


\draw (v1) -- (v2);
\draw (v2) -- (v3);
\draw (v3) -- (v4);
\draw (v2) -- (v4);

%

\end{tikzpicture}
\caption{An example of the reduction from $\BDIS{t}$ to $\ASG{t}$ from Lemma~\ref{lem:reduction_bdis_t_to_asg_t}. On the left hand side, we see the algorithm for $\BDIS{t}$ that places vertices on either level $1$ or level $2$, depending on whether it is allowed to accept the newly revealed vertex or not. On the right hand side, we see the algorithm for $\ASG{t}$ that the algorithm for $\BDIS{t}$ queries for help.}
\label{fig:bdis_t_to_asg_t}
\end{figure}

Let $\{v_{i_1},v_{i_2},\ldots,v_{i_k}\}$ be the vertices on level $1$ when $\ALG_{\BDIS{t}}$ does not receive any more requests.
Observe that for each $j=1,2,\ldots,k$, $\ALG_{\BDIS{t}}$ had the option of accepting $v_{i_j}$, and therefore it made $\ALG_{\ASG{t}}$ to guess a bit, the answer to which determined whether or not $\ALG_{\BDIS{t}}$ accepted $v_{i_j}$.

Having defined $\ALG_{\BDIS{t}}$ it only remains to define $x'$ and $\hat{x}'$ for $I'$.
We let $x_{i_j}'$ and $\hat{x}_{i_j}'$ be the true and predicted bit of the request given to $\ALG_{\ASG{t}}$ when $\ALG_{\BDIS{t}}$ received the vertex $v_{i_j}$.
For all $j=1,2,\ldots,k$ we let $\hat{x}_{i_j}' = \hat{x}_{i_j}$ and
\begin{align}\label{eq:def_of_x^f_bdis_t_to_asg_t}
x_{i_j}' = \begin{cases}
1 - x_{i_j}, &\mbox{if $x_{i_j}=y_{i_j}=1$} \\
x_{i_j}, &\mbox{otherwise.}
\end{cases}
\end{align}

Since we sometimes have that $x_{i_j} \neq x_{i_j}'$, this reduction may introduce some error. 
For later analysis, we need to bound how much error we introduce with the above reduction.

\textbf{Bounding the error of $\boldsymbol{{I'}}$:}
By definition of $I'$, the number of requests for $\ASG{t}$ may be fewer than the number of requests for $\BDIS{t}$, and so $\MU_b(I') \leqslant \MU_b(I)$, for $b \in \{0,1\}$, before taking into account new prediction errors that may be introduced in~\eqref{eq:def_of_x^f_bdis_t_to_asg_t}. 
In particular, if $x_{i_j} = y_{i_j} = 1$ then we set $x_{i_j}' = 0$, and so if $\hat{x}_{i_j} = 0$ then we change an incorrect prediction into a correct prediction, and if $\hat{x}_{i_j} = 1$ we change a correct prediction into an incorrect prediction. 
Therefore, 
\begin{align}
\MUZERO(I') &\leqslant \MUZERO(I) - \sum_{j=1}^k x_{i_j} \cdot y_{i_j} \cdot (1-\hat{x}_{i_j}) \leqslant \MUZERO(I), \label{eq:change_muzero_bdis_t_to_asg_t} \\
\MUONE(I') &\leqslant \MUONE(I) + \sum_{j=1}^k  x_{i_j} \cdot y_{i_j} \cdot \hat{x}_{i_j}. \label{eq:change_muone_bdis_t_to_asg_t}
\end{align}

This finishes the definition of the two maps $\ORALG{\rho}$ and $\ORTRANS{\rho}$ that constitute the reduction.

\textbf{Towards~\ref{item:pareto_optimal_reduction}: The existence of $\boldsymbol{\ABC}$-competitive Pareto-optimal algorithms for $\boldsymbol{\ASG{t}}$ implies the existence of $\boldsymbol{\ABC}$-competitive algorithms for $\boldsymbol{\BDIS{t}}$.}
Let $\ALG_{\ASG{t}} \in \ALGS{\ASG{t}}$ be an $\ABC$-competitive Pareto-optimal algorithm for $\ASG{t}$. 
In the following, we show that $\ORALG{\rho}(\ALG_{\ASG{t}})$ is an $\ABC$-competitive algorithm for $\BDIS{t}$, where $\ORALG{\rho}$ is the map defined above.
%
To this end, recall from Theorem~\ref{thm:pareto_optimal_algorithms_for_asg_t} that any $\ABC$-competitive Pareto-optimal algorithm for $\ASG{t}$ satisfies that $\gamma \leqslant 1$.

Given $\ALG_{\ASG{t}} \in \ALGS{\ASG{t}}$, and any instance $I \in \INSTANCES{\BDIS{t}}$, we let $\ALG_{\BDIS{t}} = \ORALG{\rho}(\ALG_{\ASG{t}})$ and $I' = \ORTRANS{\rho}(\ALG_{\ASG{t}},I)$.
We show that
\begin{align}
\alpha \cdot \OPT_{\ASG{t}}(I') + \OPT_{\BDIS{t}}(I) \leqslant\; \alpha \cdot \ALG_{\BDIS{t}}(I) + \ALG_{\ASG{t}}(I') - \sum_{j=1}^k x_{i_j} \cdot y_{i_j} \cdot \hat{x}_{i_j}. \label{eq:condition_bdis_t_to_asg_t}
\end{align}
Before verifying Equation~\eqref{eq:condition_bdis_t_to_asg_t}, we argue that the definition of $\ORALG{\rho}$ and $\ORTRANS{\rho}$ together with Equation~\eqref{eq:condition_bdis_t_to_asg_t} implies that $\ALG_{\BDIS{t}}$ is $\ABC$-competitive.

Since $\ALG_{\ASG{t}}$ is $\ABC$-competitive with additive constant $\AT$ (see Definition~\ref{def:competitiveness}), then, for any $I' \in \INSTANCES{\ASG{t}}$, we have that
\begin{align*}
-\alpha \cdot \OPT_{\ASG{t}}(I') \leqslant - \ALG_{\ASG{t}}(I') + \beta \cdot \MUZERO(I') + \gamma \cdot \MUONE(I') + \AT.
\end{align*}
Hence, by Equation~\eqref{eq:condition_bdis_t_to_asg_t} and the $\ABC$-competitiveness of $\ALG_{\ASG{t}}$, 
\begin{align*}
\OPT_{\BDIS{t}}(I) \leqslant \; &\alpha \cdot \ALG_{\BDIS{t}}(I) + \beta \cdot \MUZERO(I') + \gamma \cdot \MUONE(I') + \AT - \sum_{j=1}^k x_{i_j} \cdot y_{i_j} \cdot \hat{x}_{i_j}.
\end{align*}
Now, by Equations~\eqref{eq:change_muzero_bdis_t_to_asg_t} and~\eqref{eq:change_muone_bdis_t_to_asg_t}, and since $\ALG$ is Pareto-optimal so $\gamma \leqslant 1$,
\begin{align*}
\OPT_{\BDIS{t}}(I) \leqslant \alpha \cdot \ALG_{\BDIS{t}}(I) + \beta \cdot \MUZERO(I) + \gamma \cdot \MUONE(I) + \AT,
\end{align*}
implying that $\ALG_{\BDIS{t}}$ is indeed $\ABC$-competitive.
Therefore, to finish Part~\ref{item:pareto_optimal_reduction}, it only remains to verify that Equation~\eqref{eq:condition_bdis_t_to_asg_t} is true for any Pareto-optimal algorithm for $\ASG{t}$ and any instance of $\BDIS{t}$.

\textbf{Verifying Equation~\eqref{eq:condition_bdis_t_to_asg_t}:}
First, observe that $\ALG_{\BDIS{t}}$ can only gain profit from the vertices on level $1$. 
Therefore, 
$\ALG_{\BDIS{t}}(I) = \sum_{j=1}^k (1-y_{i_j})$.
Hence, using Equations~\eqref{obj:asg_t} and~\eqref{obj:bdis_t}, we rewrite Equation~\eqref{eq:condition_bdis_t_to_asg_t} as
\begin{align}\label{eq:new_condition_bdis_t_to_asg_t}
\underbrace{\alpha \cdot \sum_{j=1}^k x_{i_j}' + \sum_{i=1}^n (1-x_i)}_{\textsc{LHS}} \leqslant \underbrace{\alpha \cdot \sum_{j=1}^k (1-y_{i_j}) + \sum_{j=1}^k \left( y_{i_j} + t \cdot x_{i_j}' \cdot (1-y_{i_j}) \right) - \sum_{j=1}^k x_{i_j} \cdot y_{i_j} \cdot \hat{x}_{i_j}}_{\textsc{RHS}}.
\end{align}

We show that Equation~\eqref{eq:new_condition_bdis_t_to_asg_t} is satisfied at each step during the processing of $\ALG_{\BDIS{t}}$. 
Clearly, when no requests have been given to $\ALG_{\BDIS{t}}$, Equation~\eqref{eq:new_condition_bdis_t_to_asg_t} is satisfied as both LHS and RHS evaluate to $0$. 

Next, assume that Equation~\eqref{eq:new_condition_bdis_t_to_asg_t} is satisfied after $i-1$ vertices has been revealed to $\ALG_{\BDIS{t}}$, and suppose that the $i$'th vertex, $v_i$, has just been revealed.
We split the rest of this argument into six subcases:

\begin{enumerate}[label = {(\alph*)}]
\item \textit{Case $\mathit{\ell(v_i) = 1}$ and $\mathit{x_i = y_i = 1}$:} 
Since $\ell(v_i) = 1$, there exists $j \in \{1,2,\ldots,k\}$ such that $i = i_j$.
Further, as $x_{i_j} = y_{i_j} = 1$ then $x_{i_j}' = 1 - x_{i_j} = 0$. 
Therefore, LHS does not increase, and since $y_{i_j} = 1$ RHS increases by $1$ when $\hat{x}_{i_j} = 0$ and by $0$ otherwise.
Hence, the inequality in Equation~\eqref{eq:new_condition_bdis_t_to_asg_t} remains satisfied.

\item \textit{Case $\mathit{\ell(v_i) = 1}$, $\mathit{x_i = 1}$ and $\mathit{y_i = 0}$:} \label{item:cases_bdis_t_to_asg_t_110}
Since $\ell(v_i) = 1$, there is some $j \in \{1,2,\ldots,k\}$ such that $i = i_j$.
Further, since $x_{i_j} \neq y_{i_j}$, $x_{i_j}' = x_{i_j} = 1$.
Hence, LHS increases by $\alpha$.
Further, since $y_{i_j} = 0$ and $x_{i_j}' = 1$, RHS increases by $\alpha + t$ and so the inequality in Equation~\eqref{eq:new_condition_bdis_t_to_asg_t} remains satisfied.
Observe that RHS increase by $t$ more than LHS.
We will use this observation later.

\item \textit{Case $\mathit{\ell(v_i) = 1}$ and $\mathit{x_i = y_i = 0}$:} 
Since $\ell(v_i) = 1$, there exist $j \in \{1,2,\ldots,k\}$ such that $i = i_j$.
Since $x_{i_j} = 0$, then $x_{i_j}' = 0$.
Therefore, LHS increases by $1$.
Since $y_{i_j} = 0$ and $x_{i_j}' = 0$, RHS increases by $\alpha$.
Hence, the inequality in Equation~\eqref{eq:new_condition_bdis_t_to_asg_t} remains satisfied.

\item \textit{Case $\mathit{\ell(v_i) = 1}$, $\mathit{x_i = 0}$ and $\mathit{y_i = 1}$:} 
Since $\ell(v_i) = 1$, there is some $j \in \{1,2,\ldots,k\}$ such that $i = i_j$.
Since $x_{i_j} = 0$ then $x_{i_j}' = 0$.
Therefore, LHS increases by $1$.
Since $y_{i_j} = 1$, RHS increases by $1$, and so the inequality in Equation~\eqref{eq:new_condition_bdis_t_to_asg_t} remains satisfied.

\item \textit{Case $\mathit{\ell(v_i) = 2}$ and $\mathit{x_i = 1}$:} 
Since $\ell(v_i) = 2$ then $y_i = 1$, and for all $j \in \{1,2,\ldots,k\}$, $i_j \neq i$.
Since $x_i = 1$, neither LHS nor RHS increases by anything and therefore the inequality in Equation~\eqref{eq:new_condition_bdis_t_to_asg_t} remains satisfied.

\item \textit{Case $\mathit{\ell(v_i) = 2}$ and $\mathit{x_i = 0}$:} 
Since $\ell(v_i) = 2$ then $y_i = 1$, and for all $j \in \{1,2,\ldots,k\}$, $i_j \neq i$.
Since $x_i = 0$, then LHS increase by $1$, and since there does not exist $j$ such that $i_j = i$, then RHS is unchanged.

In this case, there exist a smallest $j \in \{1,2,\ldots,k\}$ with $i_j < i$ such that $(v_{i_j},v_i) \in E$ and $y_{i_j}= 0$, as otherwise $\ell(v_i)$ would have been $1$.
Since $x_i = 0$ and $(v_{i_j},v_i) \in E$, then $x_{i_j} = 1$ as otherwise $x$ encodes an infeasible solution.
Since $y_{i_j} = 0$ and $x_{i_j} = 1$, we were in case~\ref{item:cases_bdis_t_to_asg_t_110} when $v_{i_j}$ was revealed, which left RHS $t$ larger than LHS.
Therefore, we can assign the profit of $1$ gained by $\OPT_{\BDIS{t}}$ on the $i$'th request to the cost incurred by $\ALG_{\ASG{t}}$ on the $i_j$'th request.
Since $d_G(v_{i_j}) \leqslant t$, $v_{i_j}$ has at most $t$ neighbours in $G$ and so we do this accounting at most $t$ times for $v_{i_j}$, corresponding to the amount that the RHS is larger than the LHS after accounting for $v_{i_j}$.
Hence, the inequality in Equation~\eqref{eq:new_condition_bdis_t_to_asg_t} remains satisfied.
\end{enumerate}
This verifies Equation~\eqref{eq:new_condition_bdis_t_to_asg_t} and thus Equation~\eqref{eq:condition_bdis_t_to_asg_t}, which finishes Part~\ref{item:pareto_optimal_reduction}.

\textbf{Towards~\ref{item:extend_beyond_pareto_optimality}: Proving that $\boldsymbol{\ASG{t}}$ is as hard as $\boldsymbol{\BDIS{t}}$.}
Let $\ALG_{\ASG{t}} \in \ALGS{\ASG{t}}$ be any $\ABC$-competitive algorithm for $\ASG{t}$. 
If $\ALG_{\ASG{t}}$ is Pareto-optimal, Part~\ref{item:pareto_optimal_reduction} implies that $\ORALG{\rho}(\ALG_{\ASG{t}})$ is an $\ABC$-competitive algorithm for $\BDIS{t}$, and we are done.

If, on the other hand, $\ALG_{\ASG{t}}$ is not Pareto-optimal, then, by the definition of Pareto-optimality, there exists an $(\alpha',\beta',\gamma')$-competitive Pareto optimal algorithm for $\ASG{t}$, $\PAR_{\ASG{t}}$, such that $\alpha' \leqslant \alpha$, $\beta' \leqslant \beta$, and $\gamma' \leqslant \gamma$, with at least one of the inequalities being strict.
Then, by Item~\ref{item:pareto_optimal_reduction}, $\ORALG{\rho}(\PAR_{\ASG{t}})$ is an $(\alpha',\beta',\gamma')$-competitive algorithm for $\BDIS{t}$.
To see that $\ORALG{\rho}(\PAR_{\ASG{t}})$ is also $\ABC$-competitive, we observe that for any instance $I \in \INSTANCES{\BDIS{t}}$, we have that
\begin{align*}
\OPT_{\BDIS{t}}(I) &\leqslant \alpha' \cdot \ORALG{\rho}(\PAR_{\ASG{t}})(I) + \beta' \cdot \MUZERO(I) + \gamma' \cdot \MUONE(I) + \AT \\
&\leqslant \alpha \cdot \ORALG{\rho}(\PAR_{\ASG{t}})(I) + \beta \cdot \MUZERO(I) + \gamma \cdot \MUONE(I) + \AT.
\end{align*}
Hence, the existence of any $\ABC$-competitive algorithm for $\ASG{t}$ implies the existence of an $\ABC$-competitive algorithm for $\BDIS{t}$ implying that $\ASG{t}$ is as hard as $\BDIS{t}$, which finishes the proof.
\end{proof}

\subsection{$\BDIS{t}$ is as hard as $\ASG{t}$}

Next, we build a reduction in the other direction.
This reduction is inspired by the reduction template from~\cite{BBFL25}.
For this reduction, we need two technical lemmas, that imply a positive result on the competitiveness of Pareto-optimal algorithms for $\BDIS{t}$.
The first of these is inspired by a similar result for $\ASG{t}$ from~\cite{BBFL25}.

\vspace{0.2cm}
\begin{lemma}\label{lem:small_alpha_imply_large_gamma}
Let $\ALG$ be an $\ABC$-competitive algorithm for $\BDIS{t}$ with respect to $\PAIRMU$.
If $\alpha < t$, then $\gamma \geqslant 1$.
\end{lemma}
\begin{proof}
Assume towards contradiction that $\ALG$ is $\ABC$-competitive where $\alpha = t-\eps$ and $\gamma < 1$, for some $\eps > 0$.

We define two gadgets that are both small instances:
\begin{itemize}
\item $G^0_i$ contains one vertex $v_i$ with $\hat{x}_i = 1$ and $x_i = 0$.
\item $G^1_i$ contains $t+1$ vertices $v_i,w_{i,1},w_{i,2},\ldots,w_{i,t}$ and have $\hat{x}_i = x_i = 1$ and $\hat{x}_{i,j} = x_{i,j} = 0$, for all $j=1,2,\ldots,t$. 
The vertex $v_i$ is revealed first, and then the vertices $w_{i,j}$ are revealed one by one in increasing order by $j$.
\end{itemize}
Observe that $\OPT(G^0_i) = 1$ and $\OPT(G^1_i) = t$.
Further, observe that all $w_{i,j}$'s are predicted correctly, so any possible prediction error can be measured by comparing the $x_i$'s to the $\hat{x}_i$'s.
Finally, for ease of notation, we let $y_i$ be $\ALG$'s output to $v_i$. 
Observe that 
\begin{align*}
\ALG(G^0_i) = \begin{cases}
1, &\mbox{if $y_i = 0$} \\
0, &\mbox{otherwise,}
\end{cases} \hspace{0.5cm} \text{and} \hspace{0.5cm} \ALG(G^1_i) \leqslant \begin{cases}
1, &\mbox{if $y_i = 0$} \\
t, &\mbox{otherwise.}
\end{cases}
\end{align*}
We define the instance $I$, such that the $i$'th gadget is $G^0_i$, if $y_i = 1$ and $G^1_i$ otherwise. 
Observe that $I$ is well-defined since $\ALG$ is deterministic.

For the analysis, we let $Y_0$ and $Y_1$ be the number of times that $\ALG$ sets $y_i = 0$ and $y_i = 1$, respectively.
By construction of $I$, $\OPT(I) = Y_1 + t \cdot Y_0$ and $\ALG(I) = Y_0$.
Now, since $\ALG$ is $(t-\eps,\beta,\gamma)$-competitive, there exists $\AT \in \RR$ such that
\begin{align*}
Y_1 + t \cdot Y_0 &\leqslant (t-\eps) \cdot \ALG(I) + \beta \cdot \sum_{i=1}^n x_i \cdot (1-\hat{x}_i) + \gamma \cdot \sum_{i=1}^n (1-x_i) \cdot \hat{x}_i + \AT \\
&= (t-\eps) \cdot Y_0 + \gamma \cdot Y_1 + \AT.
\end{align*}
The above inequality holds if and only if
\begin{align*}
(1-\gamma) \cdot Y_1 + \eps \cdot Y_0 - \AT \leqslant 0.
\end{align*}
Finally, since $1-\gamma > 0$ as $\gamma < 1$, and $\eps > 0$, we find that
\begin{align*}
\min\{1-\gamma,\eps\} \cdot n - \AT &= \min\{1-\gamma,\eps\} \cdot (Y_1 + Y_0) - \AT \\
&\leqslant (1-\gamma) \cdot Y_1 + \eps \cdot Y_0 - \AT,
\end{align*}
and since $\lim_{n\to\infty} \min\{1-\gamma,\eps\} \cdot n - \AT = \infty$, we have a contradiction. 
\end{proof}

\vspace{0.2cm}
\begin{lemma}\label{lem:pareto-optimal-for-bdist-then-a+bleqt}
Let $\ALG$ be an $\ABC$-competitive Pareto-optimal algorithm for $\BDIS{t}$ with respect to $\PAIRMU$.
Then, $\alpha + \beta \leqslant t$.
\end{lemma}
\begin{proof}
Since $\ASG{t}$ is as hard as $\BDIS{t}$ with respect to $\PAIRMU$ by Lemma~\ref{lem:reduction_bdis_t_to_asg_t}, then the existence of a $(t,0,0)$-competitive and an $(\alpha,\beta,1)$-competitive algorithm for $\ASG{t}$ (see~\cite[Theorems 1 and 2]{BBFL25}), for all $\alpha,\beta \in \RR^+$ with $\alpha \geqslant 1$ and $\alpha + \beta = t$, implies the existence of a $(t,0,0)$-competitive and an $(\alpha,\beta,1)$-competitive algorithm for $\BDIS{t}$ with respect to $\PAIRMU$.

Now, assume towards contradiction that $\ALG$ is an $\ABC$-competitive Pareto-optimal algorithm for $\BDIS{t}$ with $\alpha + \beta > t$.

\textbf{Case~$\alpha \geqslant t$:}
Since $\alpha + \beta > t$ and $\alpha \geqslant t$, then either $\alpha > t$ or $\beta > 0$.
Then, $\ALG$ cannot be Pareto-optimal, since at least one of the three parameters can be improved without increasing any other parameter, due to the existence of the $(t,0,0)$-competitive algorithm for $\BDIS{t}$.

\textbf{Case~$\alpha < t$:}
Since $\alpha < t$ then $\gamma \geqslant 1$ by Lemma~\ref{lem:small_alpha_imply_large_gamma}.
Hence, since $\alpha + \beta > t$, the existence of an $(\alpha,t-\alpha,1)$-competitive algorithm implies that $\beta$ can be improved without increasing either $\alpha$ or $\gamma$, and so $\ALG$ cannot be Pareto-optimal.
\end{proof}


\vspace{0.2cm}
\begin{lemma}\label{lem:asg_t_to_bdis_t}
$\BDIS{t}$ is as hard as $\ASG{t}$ for all $t \in \ZZ^+$ with respect to $\PAIRMU$. 
\end{lemma}
\begin{proof}
Throughout this proof, we only measure prediction error using the canonical error measures $\PAIRMU$ defined in Equation~\eqref{eq:pairmu}. Therefore, we suppress notation and nomenclature related to error measures whenever possible.

Similarly to the proof of Lemma~\ref{lem:reduction_bdis_t_to_asg_t}, we split this proof into three parts:
\begin{enumerate}[label = {(\roman*)}]
\item \label{item:definition_of_reduction} We define two maps $\ORALG{\rho} \colon \ALGS{\BDIS{t}} \rightarrow \ALGS{\ASG{t}}$ and $\ORTRANS{\rho} \colon \ALGS{\BDIS{t}} \times \INSTANCES{\ASG{t}} \rightarrow \INSTANCES{\BDIS{t}}$, that allows us to translate algorithms for $\BDIS{t}$ into algorithms for $\ASG{t}$.
\item \label{item:correctness_of_reduction_pareto} We show that the existence of an $\ABC$-competitive Pareto-optimal algorithm for $\BDIS{t}$ implies the existence of an $\ABC$-competitive algorithm for $\ASG{t}$.
\item \label{item:extend_to_all_asg_t_algorithms} We prove that $\BDIS{t}$ is as hard as $\ASG{t}$.
\end{enumerate}

\textbf{Towards~\ref{item:definition_of_reduction}: Definition of $\boldsymbol{{\ORALG{\rho}}}$ and $\boldsymbol{{\ORTRANS{\rho}}}$.}
Let $\ALG_{\BDIS{t}} \in \ALGS{\BDIS{t}}$ and $I = (x,\hat{x},r) \in \INSTANCES{\ASG{t}}$.
Then, we define $\ALG_{\ASG{t}} = \ORALG{\rho}(\ALG_{\BDIS{t}})$ and $I' = (x',\hat{x}',r') = \ORTRANS{\rho}(\ALG_{\BDIS{t}},I)$ as follows.

When $\ALG_{\ASG{t}}$ is asked to guess the bit $x_i$ given the prediction $\hat{x}_i$, it gives a new request containing an isolated vertex, a \emph{challenge request}, $v_i'$, with prediction $\hat{x}_i' = \hat{x}_i$ to $\ALG_{\BDIS{t}}$.
Then, $\ALG_{\ASG{t}}$ outputs the same guess of $x_i$ as $\ALG_{\BDIS{t}}$ outputs on $v_i'$. 

When $\ALG_{\ASG{t}}$ receives no more requests, and therefore learns the true contents of $x$, it determines $x'$, by setting
\begin{align}\label{eq:definition_of_x^f_asgt_to_bdist}
x_i' = \begin{cases}
1-x_i, &\mbox{if $x_i = y_i = 0$,}\\
x_i, &\mbox{otherwise,}
\end{cases}
\end{align} 
where $x_i'$ is the true bit of $v_i'$.
By construction, there is no challenge request with $x_i' = y_i = 0$.

After computing $x'$, $\ALG_{\ASG{t}}$ adds a block of requests for each challenge request given to $\ALG_{\BDIS{t}}$.
We let $B_{v_i'}$ be the block of requests associated to the challenge request, $v_i'$.
The blocks are defined as follows.
\begin{enumerate}[label = {(\alph*)}]
\item \textit{Case $\mathit{x_i' = 1}$ and $\mathit{y_i = 0}$:} \label{block:10}
Request $t$ new vertices, $v_{i,j}'$, for $j=1,2,\ldots,t$, where $v_{i,j}'$ is connected to $v_i'$ with prediction $\hat{x}_{i,j}' = 0$.
Clearly, $\OPT_{\BDIS{t}}$ accepts all $v_{i,j}'$, and $\ALG_{\BDIS{t}}$ must reject all $v_{i,j}'$ to avoid creating an infeasible solution.
Hence, $x_{i,j}' = 0$ and $y_{i,j} = 1$ for all $j=1,2,\ldots,t$. 

\item \textit{Case $\mathit{x_i' = 0}$ and $\mathit{y_i = 1}$:}\label{block:01}
This block is empty. 
Since $x_i' = 0$, and $v_i'$ is isolated, it is optimal to accept $v_i'$ into the independent set.

\item \textit{Case $\mathit{x_i' = y_i = 1}$:} \label{block:11}
Request a new vertex, $v_{i,1}'$, that is connected to $v_i$ with $x_{i,1}' = \hat{x}_{i,1}' = 0$.
Observe that $v_i'$ is necessary to ensure that $x'$ encodes an optimal solution.
$\ALG_{\BDIS{t}}$ may either accept or reject $v_{i,1}'$.
\end{enumerate}

Similarly to the proof of Lemma~\ref{lem:reduction_bdis_t_to_asg_t}, we sometimes have that $x_i' \neq x_i$, meaning that the error of $I'$ may differ from that of $I$.
For later analysis, we need to bound how much error we introduce with the above reduction.

\textbf{Bounding the error of $\boldsymbol{{I'}}$:}
By Equation~\eqref{eq:definition_of_x^f_asgt_to_bdist}, the true bits of some challenge requests for $\BDIS{t}$ are different from the true bits of the corresponding requests for $\ASG{t}$, while the predicted bits of the same requests coincide.
In particular, if $x_i = y_i = 0$, then $x_i' = 1$.
Hence, if $x_i = y_i = 0$ and $\hat{x}_i' = \hat{x}_i = 1$ then we change an incorrect prediction into a correct prediction, and if $x_i = y_i = 0$ and $\hat{x}_i' = \hat{x}_i = 0$ then we change a correct prediction into an incorrect prediction.
In the following, we bound the error of $I'$ as a function of the error of $I$:
\begin{align}
\MUZERO(I') &= \MUZERO(I) + \sum_{i=1}^{n} (1-x_i) \cdot (1-y_i) \cdot (1-\hat{x}_i) \label{eq:muzero_of_is_instance} \\
\MUONE(I') &= \MUONE(I) - \sum_{i=1}^{n} (1-x_i) \cdot (1-y_i) \cdot \hat{x}_i \leqslant \MUONE(I). \label{eq:muone_of_is_instance}
\end{align}

This finishes the definition of the two maps $\ORALG{\rho}$ and $\ORTRANS{\rho}$ that constitute the reduction.

\textbf{Towards~\ref{item:correctness_of_reduction_pareto}: The existence of an $\boldsymbol{\ABC}$-competitive Pareto-optimal algorithm for $\boldsymbol{\BDIS{t}}$ implies the existence of an $\boldsymbol{\ABC}$-competitive algorithm for $\boldsymbol{\ASG{t}}$.}
Let $\ALG_{\BDIS{t}} \in \ALGS{\BDIS{t}}$ be an $\ABC$-competitive Pareto-optimal algorithm for $\BDIS{t}$.
We show that $\ALG_{\ASG{t}} = \ORALG{\rho}(\ALG_{\BDIS{t}})$ is an $\ABC$-competitive algorithm for $\ASG{t}$, where $\ORALG{\rho}$ is the map defined above. 
For this, recall from Lemma~\ref{lem:pareto-optimal-for-bdist-then-a+bleqt} that $\ABC$-competitive Pareto-optimal algorithms for $\BDIS{t}$ satisfy that $\alpha + \beta \geqslant t$.

Observe that each request, $r_i$, in $I$ induces a subgraph $H_i'$ in the underlying graph of $I'$.
By construction, $H_i'$ contains the vertex $v_i'$ and the vertices and edges from $B_{v_i'}$.
We let $\ALG_{\BDIS{t}}(H_i')$ and $\OPT_{\BDIS{t}}(H_i')$ be the profit of $\ALG_{\BDIS{t}}$ and $\OPT_{\BDIS{t}}$ on the requests in $I'$ corresponding to $H_i'$. 
Further, observe that for any $i,j \in \{1,2,\ldots,n\}$ with $i \neq j$, $H_i'$ and $H_j'$ are not connected by any edges, and so $\ALG_{\BDIS{t}}(H_i' \cup H_j') = \ALG_{\BDIS{t}}(H_i') + \ALG_{\BDIS{t}}(H_j')$, for any $\ALG \in \ALGS{\BDIS{t}}$.
The same is true for $\OPT_{\BDIS{t}}$. 

Next, we explain the structure of the rest of the proof.
We show that there exists a function $d \colon \{0,1\} \times \{0,1\} \rightarrow \ZZ^+$ such that for all $x_i,y_i \in \{0,1\}$,
\begin{equation}\label{eq:condition_alg_arbitrary_consistency} 
y_i + t \cdot x_i \cdot (1-y_i) \leqslant d(x_i,y_i) - \alpha \cdot \ALG_{\BDIS{t}}(H_i') - \beta \cdot (1-x_i) \cdot (1-y_i) \cdot (1-\hat{x}_i) \tag{C1} 
\end{equation}
and
\begin{align}\label{eq:condition_opt_arbitrary_consistency}
d(x_i,y_i) - \OPT_{\BDIS{t}}(H_i') \leqslant \alpha \cdot x_i. \tag{C2} 
\end{align}
Before proving the existence of $d$, we argue that 
the existence of $d$ implies that if $\ALG_{\BDIS{t}}$ is an $\ABC$-competitive for $\BDIS{t}$ with respect to $\PAIRMU$, then $\ALG_{\ASG{t}} = \ORALG{\rho}(\ALG_{\BDIS{t}})$ is an $\ABC$-competitive for $\ASG{t}$ with respect to $\PAIRMU$.
To this end, assume that $\ALG_{\BDIS{t}}$ is $\ABC$-competitive, let $\AT$ be the additive constant of $\ALG_{\BDIS{t}}$, let $I \in \INSTANCES{\ASG{t}}$, and let $I' = \ORTRANS{\rho}(\ALG_{\BDIS{t}},I)$.
Then, by Equations~\eqref{obj:asg_t} and~\eqref{eq:condition_alg_arbitrary_consistency}, we have that
\begin{align*}
\ALG_{\ASG{t}}(I) = \; \sum_{i=1}^n \left( y_i + t \cdot x_i \cdot (1-y_i) \right) \leqslant \; \sum_{i=1}^n \left( d(x_i,y_i) - \alpha \cdot \ALG_{\BDIS{t}}(H_i') - \beta \cdot (1-x_i) \cdot (1-y_i) \cdot (1-\hat{x}_i) \right).
\end{align*}
Further, by Equation~\eqref{eq:muzero_of_is_instance} and since $\sum_{i=1}^n \ALG_{\BDIS{t}}(H_i') = \ALG_{\BDIS{t}}(I')$, we get that
\begin{align*}
\ALG_{\ASG{t}}(I) \leqslant \left(\sum_{i=1}^n d(x_i,y_i)\right) - \alpha \cdot \ALG_{\BDIS{t}}(I') - \beta \cdot (\MUZERO(I') - \MUZERO(I)).
\end{align*}
Using the $\ABC$-competitiveness of $\ALG_{\BDIS{t}}$ (see Definition~\ref{def:competitiveness}), we get that
\begin{align*}
\ALG_{\ASG{t}}(I) \leqslant 
\; &\left(\sum_{i=1}^n d(x_i,y_i)\right) - \OPT_{\BDIS{t}}(I') + \beta \cdot \MUZERO(I) + \gamma \cdot \MUONE(I') + \AT.
\end{align*}
Now, by Equation~\eqref{eq:muone_of_is_instance} and since $\OPT_{\BDIS{t}}(I') = \sum_{i=1}^n \OPT_{\BDIS{t}}(H_i')$, we have
\begin{align*}
\ALG_{\ASG{t}}(I) \leqslant \left(\sum_{i=1}^n d(x_i,y_i) - \OPT_{\BDIS{t}}(H_i')\right) + \beta \cdot \MUZERO(I) + \gamma \cdot \MUONE(I) + \AT.
\end{align*}
Finally, by Equations~\eqref{eq:condition_opt_arbitrary_consistency} and~\eqref{obj:asg_t}, we get that
\begin{align*}
\ALG_{\ASG{t}}(I) \leqslant\; &\alpha \cdot \OPT_{\ASG{t}}(I) + \beta \cdot \MUZERO(I) + \gamma \cdot \MUONE(I) + \AT.
\end{align*}
Hence, if $\ALG_{\BDIS{t}}$ is an $\ABC$-competitive algorithm with respect to $\PAIRMU$, then the existence of $d$ is sufficient to show that $\ALG_{\ASG{t}}$ is an $\ABC$-competitive algorithm for $\ASG{t}$ with respect to $\PAIRMU$.

For now, we focus on preserving the competitiveness of Pareto-optimal algorithms for $\BDIS{t}$.
In this setting, recall that any $\ABC$-competitive Pareto-optimal algorithm for $\BDIS{t}$ satisfies that $\alpha + \beta \leqslant t$, by Lemma~\ref{lem:pareto-optimal-for-bdist-then-a+bleqt}.
Hence, proving the existence of a function $d$ satisfying~\eqref{eq:condition_alg_arbitrary_consistency} and~\eqref{eq:condition_opt_arbitrary_consistency} whenever $\alpha + \beta \leqslant t$ will be sufficient to finish the proof of Item~\ref{item:correctness_of_reduction_pareto}.

\textbf{The existence of $\boldsymbol{{d}}$:}
We show that the function
\begin{align}\label{eq:proposed_d}
d(x_i,y_i) =  \abs{H_i'} - (1-x_i) \cdot (1-y_i) + (\alpha - 1) \cdot x_i
\end{align}
where $\smallabs{H_i'}$ is the number of vertices in $H_i'$, satisfies Equations~\eqref{eq:condition_alg_arbitrary_consistency} and~\eqref{eq:condition_opt_arbitrary_consistency} whenever $\alpha + \beta \leqslant t$.
We consider four cases:
\begin{description}
\item \textit{Case $\mathit{x_i = 1}$ and $\mathit{y_i = 0}$:}
In this case $x_i' = 1$ and so $H_i$ contains $v_i'$ and the $t$ vertices in the block $B_{v_i'}$ (see~\ref{block:10}).
Therefore, $\smallabs{H_i'} = t+1$, $\OPT_{\BDIS{t}}(H_i') = t$, and $\ALG_{\BDIS{t}}(H_i') = 1$, and so $d(1,0) = t+1 - 0 + \alpha - 1 = t + \alpha$. 
Hence, Equation~\eqref{eq:condition_alg_arbitrary_consistency} is satisfied as
\begin{align*}
t \leqslant t + \alpha - \alpha - 0 = t,
\end{align*}
and Equation~\eqref{eq:condition_opt_arbitrary_consistency} is satisfied as
\begin{align*}
t + \alpha - t \leqslant \alpha.
\end{align*}
\item \textit{Case $\mathit{x_i = 0}$ and $\mathit{y_i = 1}$:}
In this case $x_i' = 0$ and so $H_i'$ only contains $v_i'$ as $B_{v_i'}$ is empty (see~\ref{block:01}).
Therefore, $\smallabs{H_i'} = 1$. $\OPT_{\BDIS{t}}(H_i') = 1$, and $\ALG_{\BDIS{t}}(H_i') = 0$, and so $d(0,1) = 1 - 0 + 0 = 1$.
Hence, Equation~\eqref{eq:condition_alg_arbitrary_consistency} is satisfied as
\begin{align*}
1 \leqslant 1 - 0 - 0,
\end{align*}
and Equation~\eqref{eq:condition_opt_arbitrary_consistency} is satisfied as
\begin{align*}
1 - 1 \leqslant 0.
\end{align*}
\item \textit{Case $\mathit{x_i = y_i = 1}$:}
In this case, $x_i' = 1$ and so $H_i'$ contains the two vertices $v_i'$ and $v_{i,1}'$ (see~\ref{block:11}).
Therefore, $\smallabs{H_i'} = 2$, $\OPT_{\BDIS{t}}(H_i') = 1$, and $\ALG_{\BDIS{t}}(H_i') \in \{0,1\}$.
Observe that $\ALG_{\BDIS{t}}(H_i') \in \{0,1\}$ as $\ALG_{\BDIS{t}}$ may either accept or reject $v_{i,1}'$.
Hence, $d(1,1) = 2 - 0 + \alpha - 1 = \alpha + 1$, and so Equation~\eqref{eq:condition_alg_arbitrary_consistency} is satisfied as
\begin{gather*}
1 \leqslant \alpha + 1 - \alpha - 0 = 1 \hspace{1cm} (\mbox{$\ALG_{\BDIS{t}}$ accepts $v_{i,1}$}) \\
1 \leqslant \alpha + 1 - 0 - 0 = 1 + \alpha \hspace{1cm} (\mbox{$\ALG_{\BDIS{t}}$ rejects $v_{i,1}$})
\end{gather*}
and~\eqref{eq:condition_opt_arbitrary_consistency} is satisfied as:
$\alpha + 1 - 1 \leqslant \alpha$.
\item \textit{Case $\mathit{x_i = y_i = 0}$:}
In this case $x_i' = 1$ (see Equation~\eqref{eq:definition_of_x^f_asgt_to_bdist}), and so $H_i'$ contains the vertex $v_i'$ and the $t$ vertices in the block $B_{v_i'}$ (see~\ref{block:10}).
Therefore, $\smallabs{H_i'} = t+1$, $\OPT_{\BDIS{t}}(H_i') = t$, and $\ALG_{\BDIS{t}}(H_i') = 1$, and so $d(0,0) = t + 1 - 1 + 0 = t$.
Hence Equation~\eqref{eq:condition_alg_arbitrary_consistency} is satisfied as
\begin{gather*}
0 \leqslant t - \alpha - 0 \hspace{1cm} (\mbox{if $\hat{x}_i = 1$})\\
0 \leqslant t - \alpha - \beta \hspace{1cm} (\mbox{if $\hat{x}_i = 0$})
\end{gather*}
Observe that both inequalities are satisfied as $\alpha + \beta \leqslant t$. 
Further,~\eqref{eq:condition_opt_arbitrary_consistency} is satisfied as
$t - t \leqslant 0$.
\end{description}

Hence, the function $d$ proposed in Equation~\eqref{eq:proposed_d} satisfies Equation~\eqref{eq:condition_alg_arbitrary_consistency} and~\eqref{eq:condition_opt_arbitrary_consistency} for all $x_i,y_i \in \{0,1\}$, whenever $\alpha + \beta \leqslant t$.
This finishes the proof of Item~\ref{item:correctness_of_reduction_pareto}.

\textbf{Towards~\ref{item:extend_to_all_asg_t_algorithms}: Proving that $\boldsymbol{\BDIS{t}}$ is as hard as $\boldsymbol{\ASG{t}}$.}
To show that $\BDIS{t}$ is as hard as $\ASG{t}$, we need to preserve the competitiveness of all algorithms for $\BDIS{t}$, and not only the Pareto-optimal ones. 
This can be done using the same strategy as in Item~\ref{item:extend_beyond_pareto_optimality} in the proof of Lemma~\ref{lem:reduction_bdis_t_to_asg_t}.
\end{proof}

We are now ready to state the main result of this section:

\vspace{0.2cm}
\begin{theorem}\label{thm:main_result}
$\BDIS{t}$ is $\CCWM{\PAIRMUFORCC}{t}$-complete for all $t \in \ZZ^+$.
\end{theorem}

By a slight modification to the proofs of Lemmas~\ref{lem:reduction_bdis_t_to_asg_t} and~\ref{lem:asg_t_to_bdis_t}, we get the following corollary, which extends the result of Theorem~\ref{thm:main_result} to the purely online setting:

\vspace{0.2cm}
\begin{corollary}
$\BDIS{t}$ is $\CCWM{\PAIRZMFORCC}{t}$-complete for all $t \in \ZZ^+$.
\end{corollary}

The main motivation for comparing the hardness of $\BDIS{t}$ to the hardness of $\ASG{t}$ is that it extends to a hardness comparison between $\BDIS{t}$ and \emph{Online $t$-Bounded Degree Vertex Cover with Predictions} ($\BDVC{t}$), the dual problem of $\BDIS{t}$.
In particular, by~\cite{BBFL25}, it is known that $\BDVC{t}$ is $\CCWM{\PAIRMUFORCC}{t}$-complete, and so, by the $\CCWM{\PAIRMUFORCC}{t}$-completeness of $\BDIS{t}$, there exists an $\ABC$-competitive algorithm for $\BDVC{t}$ if and only if there exists an $\ABC$-competitive algorithm for $\BDIS{t}$.
This is rather surprising as, in the context of the related field of approximation algorithms, it is well-known that (Bounded Degree) Vertex Cover is much easier to approximate than (Bounded Degree) Independent Set~\cite{CLRS_v4,H99}.

\section{Missing Proofs from Section~\ref{sec:a_collection_of_problems_from_CCWM_max}}\label{sec:missing_proofs}

In Section~\ref{sec:a_collection_of_problems_from_CCWM_max}, we were not able to give the proofs of Lemmas~\ref{lem:bdis_t_to_sch_t} and~\ref{lem:more_colors_is_easier}, since we needed some knowledge about which algorithms exist for $\BDIS{t}$ and $\SCH{t}$ that was not available before.
However, the result from Theorem~\ref{thm:main_result} combined with the result from Lemma~\ref{lem:sch_t_to_bdis_t} gives a much deeper understanding of this.
Therefore, we are now able to give the proofs of Lemmas~\ref{lem:bdis_t_to_sch_t} and~\ref{lem:more_colors_is_easier} in full detail.

\subsection{Proving Lemma~\ref{lem:bdis_t_to_sch_t}}\label{sec:proof_of_lemma_bdis_t_to_sch_t}

The proof of Lemma~\ref{lem:bdis_t_to_sch_t} is similar to that of Lemma~\ref{lem:reduction_bdis_t_to_asg_t}.
To complete the proof, we need a result similar to that of Theorem~\ref{thm:pareto_optimal_algorithms_for_asg_t} to apply for algorithms for $\SCH{t}$.

\vspace{0.2cm}
\begin{lemma}\label{lem:negative_result_scht}
Let $\ALG$ be an $\ABC$-competitive algorithm for $\SCH{t}$ with respect to $\PAIRMU$.
Then, $\alpha + \beta \geqslant t$.
\end{lemma}
\begin{proof}
Assume towards contradiction that $\ALG \in \ALGS{\SCH{t}}$ is $\ABC$-competitive with $\alpha + \beta < t$, and consider the family of instances $\{I^n = (x^n,\hat{x}^n,r^n)\}_{n \in \ZZ^+} \subset \INSTANCES{\SCH{t}}$ defined as follows.

The first $n$ requests in $I^n$ are non-overlapping intervals of length $t$, i.e., for $1 \leqslant i \leqslant n$ we let $r^n_i = ((i-1) \cdot t, i \cdot t)$.
The prediction for $r^n_i$ is $\hat{x}^n_i = 0$, and the true bit is $x^n_i = 1 - y_i$, where $y_i$ is $\ALG$'s response to $r^n_i$.
After these first $n$ request, we give $t$ more requests for each interval, $r^n_i$, where $y_i = 0$ (i.e., when $x^n_i = 1$).
Specifically, if $y_i = 0$, we request the $t$ intervals $((i-1) \cdot t, (i-1) \cdot t + 1), ((i-1) \cdot t + 1, (i-1) \cdot t + 2), \ldots, ((i-1) \cdot t + t-1, i \cdot t)$.
All of these intervals have true and predicted bits $0$.

By construction, 
\begin{align*}
\ALG(I^n) = \sum_{i=1}^n (1-y_i), \hspace{0.5cm} \mbox{ and } \hspace{0.5cm} \OPT(I^n) = \left(\sum_{i=1}^n y_i \right) +  t \cdot \left( \sum_{i=1}^n (1-y_i) \right).
\end{align*}
Now, since $\ALG$ is $\ABC$-competitive there exists some $\AT \in \RR^+$ such that
\begin{align*}
\OPT(I^n) \leqslant \alpha \cdot \ALG(I^n) + \beta \cdot \MUZERO(I^n) + \gamma \cdot \MUONE(I^n) + \AT = \alpha \cdot \ALG(I^n) + \beta \cdot \MUZERO(I^n) + \AT.
\end{align*}
Since $\MUZERO(I^n) = \sum_{i=1}^n x^n_i \cdot (1-\hat{x}^n_i) = \sum_{i=1}^n (1-y_i)$, we can rewrite this as
\begin{align*}
\left(\sum_{i=1}^n y_i \right) +  t \cdot \left( \sum_{i=1}^n (1-y_i) \right) \leqslant \alpha \cdot \sum_{i=1}^n (1-y_i) + \beta \cdot \sum_{i=1}^n (1-y_i) + \AT,
\end{align*}
which can be rewritten as
\begin{align*}
\left(\sum_{i=1}^n y_i \right) +  (t - (\alpha + \beta)) \cdot \left( \sum_{i=1}^n (1-y_i) \right) \leqslant \AT.
\end{align*}
Now, observe that $(t-(\alpha + \beta)) \cdot n$ is dominated by the left hand side of the above expression.
Moreover, since $\alpha + \beta < t$, then $t-(\alpha + \beta) > 0$, meaning that $(t-(\alpha+\beta)) \cdot n > 0$.
Since $\lim_{n\to\infty} (t-(\alpha+\beta)) \cdot n = \infty$, we cannot have that $(t-(\alpha + \beta)) \cdot n \leqslant \AT$ for all $n \in \ZZ^+$, since $\AT \in \RR^+$ is a constant independent of $n$, a contradiction.
\end{proof}

\vspace{0.2cm}
\begin{lemma}\label{lem:pareto-optimal_for_sch_t}
Any $\ABC$-competitive Pareto-optimal algorithm for $\SCH{t}$ with respect to $\PAIRMU$ has $\gamma \leqslant 1$.
\end{lemma}
\begin{proof}
Assume towards contradiction that $\ALG \in \ALGS{\SCH{t}}$ is an $\ABC$-competitive Pareto-optimal algorithm for $\SCH{t}$ with respect to $\PAIRMU$ with $\gamma > 1$.
By Lemma~\ref{lem:negative_result_scht}, we have that $\alpha + \beta \geqslant t$, meaning that $\beta \geqslant t - \alpha$.

Now, by Lemmas~\ref{lem:reduction_bdis_t_to_asg_t} and~\ref{lem:sch_t_to_bdis_t}, and the transitivity of the as hard as relation, it follows that $\ASG{t}$ is as hard as $\SCH{t}$ with respect to $\PAIRMU$.
By~\cite[Theorem 2]{BBFL25}, there exists an $(\alpha,t-\alpha,1)$-competitive algorithm for $\ASG{t}$ for all $t \in \ZZ^+$ and all $t \geqslant \alpha > 1$.
Since $\ASG{t}$ is as hard as $\SCH{t}$, there also exist an $(\alpha,t-\alpha,1)$-competitive algorithm for $\SCH{t}$.

However, the existence of an $(\alpha , t-\alpha,1)$-competitive algorithm for $\SCH{t}$ contradicts the Pareto-optimality of $\ALG$, since we can lower $\gamma$ to $1$ without making $\alpha$ or $\beta$ larger.
\end{proof}

We are now ready to prove Lemma~\ref{lem:bdis_t_to_sch_t}:

\vspace{0.2cm}
\bdisttoscht*
\begin{proof}
Throughout the proof, we only measure prediction error using the canonical error measures $\PAIRMU$ defined in Equation~\eqref{eq:pairmu}. 
Therefore, we suppress notation and nomenclature related to error measures whenever possible.

We define a tuple $\rho = \ONLRED{\rho}$ that closely resembles a max-reduction from $\BDIS{t}$ to $\SCH{t}$, without actually being one.
Therefore, similarly to previous proofs, we split this proof into several parts:
\begin{enumerate}[label = {(\roman*)}]
\item \label{item:define_maps_new} Define two maps $\ORALG{\rho} \colon \ALGS{\SCH{t}} \rightarrow \ALGS{\BDIS{t}}$ and $\ORTRANS{\rho} \colon \ALGS{\SCH{t}} \times \INSTANCES{\BDIS{t}} \rightarrow \INSTANCES{\SCH{t}}$.
\item \label{item:pareto_optimal_reduction_new} We show that if $\ALG_{\SCH{t}} \in \ALGS{\SCH{t}}$ is an $\ABC$-competitive Pareto-optimal algorithm for $\SCH{t}$, then $\ALG_{\BDIS{t}} = \ORALG{\rho}(\ALG_{\SCH{t}})$ is an $\ABC$-competitive for $\BDIS{t}$.
\item \label{item:extend_beyond_pareto_optimality_new} We show that $\SCH{t}$ is as hard as $\BDIS{t}$.
\end{enumerate}

\textbf{Towards~\ref{item:define_maps_new}: Definition of $\boldsymbol{{\ORALG{\rho}}}$ and $\boldsymbol{{\ORTRANS{\rho}}}$.}
Let $\ALG_{\SCH{t}} \in \ALGS{\SCH{t}}$ be any algorithm for $\SCH{t}$ and $I \in \INSTANCES{\BDIS{t}}$ be any instance of $\BDIS{t}$.
Then, we define $I' = (x',\hat{x}',r') = \ORTRANS{\rho}(\ALG_{\SCH{t}},I)$ and $\ALG_{\BDIS{t}} = \ORALG{\rho}(\ALG_{\SCH{t}})$ as follows.

Let $G = (V,E)$ be the underlying graph of $I$ and $G_i = (\{v_1,v_2,\ldots,v_i\},E_1 \cup E_2 \cup \cdots \cup E_i)$ be the subgraph of $G$ that is available to $\ALG_{\BDIS{t}}$ after $i$ requests have been given.

When $\ALG_{\BDIS{t}}$ receives a request containing a vertex, $v_i$, together with a collection of edges, $E_i$, and a prediction, $\hat{x}_i$, it determines the \emph{level} of $v_i$, denoted $\ell(v_i)$.
If $v_i$ has a neighbour, $v_j$, in $G_i$ such that $y_j = 0$, then $\ALG_{\BDIS{t}}$ sets $\ell(v_i) = 2$ and rejects $v_i$.
If, on the other hand, all neighbours of $v_i$ in $G_i$ has been rejected from $\ALG_{\BDIS{t}}$'s independent set, then it sets $\ell(v_i) = 1$, and gives a challenge request to $\ALG_{\SCH{t}}$.
The challenge request consists of the new open interval $r_i' = (a,a+t)$ with prediction $\hat{x}_i' = \hat{x}_i$, where $a$ is used to ensure that all challenge requests are non-overlapping. 
Initially, $a = 0$ and then we update $a$ to $a + t$ every time a new challenge request has been given to $\ALG_{\SCH{t}}$.
Finally, $\ALG_{\BDIS{t}}$ accepts $v_i$ if and only if $\ALG_{\SCH{t}}$ accepts $r_i'$. 

Let $\{v_{i_1},v_{i_2},\ldots,v_{i_k}\}$ be the vertices on level $1$ when $\ALG_{\BDIS{t}}$ does not receive any further request.
These are the vertices for which we gave a challenge request to $\ALG_{\SCH{t}}$, meaning that the intervals that has been requested for $\SCH{t}$ are $\{r_{i_1}',r_{i_2}',\ldots,r_{i_k}'\}$.
To align notation, we let $x_{i_j}'$ and $\hat{x}_{i_j}'$ be the true and predicted bits of $r_{i_j}'$.
For each $j=1,2,\ldots,k$, we let $\hat{x}_{i_j}' = \hat{x}_{i_j}$ and
\begin{align}\label{eq:true_bits_challenge_requests_bdis_t_to_sch_t}
x_{i_j}' = \begin{cases}
1-x_{i_j}, &\mbox{if $x_{i_j} = y_{i_j} = 1$} \\
x_{i_j}, &\mbox{otherwise.}
\end{cases}
\end{align}
Then, for each interval $r_{i_j}' = (a_{i_j},b_{i_j})$ with $x_{i_j}' = 1$, $\ALG_{\BDIS{t}}$ gives $t$ further requests to $\ALG_{\SCH{t}}$.
In particular, for $l=1,2,\ldots,t$ it gives the request $r_{i_j,l}' = (a_{i_j} + (l-1) ,b_{i_j} + l)$ with true and predicted bit $\hat{x}_{i_j,l}' = x_{i_j,l}' = 0$. 
This finishes the instance for $\SCH{t}$ and thus the definition of $I'$ and $\ALG_{\BDIS{t}}$. 

Before starting the analysis, we prove that $I'$ is a valid instance of $\SCH{t}$, and we bound the error that is introduced with this reduction.

\textbf{Showing that $\boldsymbol{{I' \in \INSTANCES{\SCH{t}}}}$:}
We prove that $I'$ is a valid instance for $\SCH{t}$.
To this end, we show that $x'$ encodes an optimal solution to $I'$, and we show that no requested interval overlaps more than $t$ other intervals.
If $x_{i_j}' = 0$ then the interval $s_{i_j}$ does not intersect any other and so it is optimal to include this interval.
If, on the other hand, $x_{i_j}' = 1$ then $s_{i_j}$ intersects $t$ other intervals, $s_{i_j,l}'$, and so it is optimal to reject $s_{i_j}$ and instead accept all $s_{i_j,l}$, all of whom has true bit $x_{i_j,l}' = 0$. 
Hence, $x'$ encodes an optimal solution to $I'$.
Finally, by construction no interval in $I'$ overlaps more than $t$ other intervals, and so $I' \in \INSTANCES{\SCH{t}}$. 

\textbf{Bounding the error of $\boldsymbol{{I'}}$:}
By definition of $x'$ in Equation~\eqref{eq:true_bits_challenge_requests_bdis_t_to_sch_t}, the prediction error of $I'$ may be different from the prediction error of $I$, since we sometimes change the true bit of a request without changing the prediction.
In particular, if $x_{i_j} = y_{i_j} = 1$ then we set $x_{i_j}' = 0$ such that $x_{i_j} \neq x_{i_j}'$.
Hence, if $\hat{x}_{i_j} = 0$ then we change a wrong prediction to a correct prediction, and if $\hat{x}_{i_j} = 1$ we change a correct prediction into a wrong prediction.
Hence, when $\hat{x}_{i_j} = 0$ we reduce the value of $\MUZERO$ by $1$, and when $\hat{x}_{i_j} = 1$ we increase the value of $\MUONE$ by $1$.
Since we do not introduce any prediction error outside the challenge requests, we get that
\begin{align}
\MUZERO(I') &\leqslant \MUZERO(I) - \sum_{j=1}^k x_{i_j} \cdot y_{i_j} \cdot (1-\hat{x}_{i_j}) \leqslant \MUZERO(I) \label{eq:bdis_t_to_sch_t_muzero_change} \\
\MUONE(I') &\leqslant \MUONE(I) + \sum_{j=1}^k x_{i_j} \cdot y_{i_j} \cdot \hat{x}_{i_j}. \label{eq:bdis_t_to_sch_t_muone_change}
\end{align}

\textbf{Towards~\ref{item:pareto_optimal_reduction_new}: We show that $\boldsymbol{\ABC}$-competitive Pareto-optimal algorithms for $\SCH{t}$ translates into $\boldsymbol{\ABC}$-competitive algorithms for $\BDIS{t}$.}
Let $\ALG_{\SCH{t}} \in \ALGS{\SCH{t}}$ be an $\ABC$-competitive Pareto-optimal algorithm for $\SCH{t}$.
In the following, we show that $\ALG_{\BDIS{t}} = \ORALG{\rho}(\ALG_{\SCH{t}})$ is an $\ABC$-competitive algorithm for $\BDIS{t}$, where $\ORALG{\rho}$ is the map defined in Item~\ref{item:define_maps_new}.
To this end, recall from Lemma~\ref{lem:pareto-optimal_for_sch_t} that $\gamma \leqslant 1$.

\textbf{Structure of the analysis:}
We show that for all instances $I \in \INSTANCES{\BDIS{t}}$ with $I' = \ORTRANS{\rho}(\ALG_{\SCH{t}},I)$, we have that
\begin{enumerate}[label = {(\alph*)}]
\item $\ALG_{\SCH{t}}(I') \leqslant \ALG_{\BDIS{t}}(I)$, and \label{item:part_1_bdis_t_to_sch_t}
\item $\OPT_{\BDIS{t}}(I) \leqslant \OPT_{\SCH{t}}(I') - \sum_{j=1}^k x_{i_j} \cdot y_{i_j} \cdot \hat{x}_{i_j}$. \label{item:part_2_bdis_t_to_sch_t}
\end{enumerate}
We argue that these conditions are sufficient to prove that $\ALG_{\BDIS{t}}$ is $\ABC$-competitive with respect to $\PAIRMU$ before proving that they are satisfied.
To this end, assume that~\ref{item:part_1_bdis_t_to_sch_t} and~\ref{item:part_2_bdis_t_to_sch_t} are satisfied, and assume that $\ALG_{\SCH{t}}$ is an $\ABC$-competitive Pareto-optimal algorithm for $\SCH{t}$ with additive constant $\AT$.
First, by~\ref{item:part_2_bdis_t_to_sch_t}, we get that
\begin{equation*}
\OPT_{\BDIS{t}}(I) \leqslant \OPT_{\SCH{t}}(I') - \sum_{k=1}^k x_{i_j} \cdot y_{i_j} \cdot \hat{x}_{i_j}.
\end{equation*}
Now, using the $\ABC$-competitiveness of $\ALG_{\SCH{t}}$ (see Definition~\ref{def:competitiveness}), then
\begin{equation*}
\OPT_{\BDIS{t}}(I) \leqslant  \alpha \cdot \ALG_{\SCH{t}}(I') + \beta \cdot \MUZERO(I') + \gamma \cdot \MUONE(I') + \AT - \sum_{j=1}^k x_{i_j}\cdot y_{i_j} \cdot \hat{x}_{i_j}.
\end{equation*}
Then, by~\ref{item:part_1_bdis_t_to_sch_t} and Equations~\eqref{eq:bdis_t_to_sch_t_muzero_change} and~\eqref{eq:bdis_t_to_sch_t_muone_change},
\begin{equation*}
\OPT_{\BDIS{t}}(I) \leqslant 
\alpha \cdot \ALG_{\BDIS{t}}(I) + \beta \cdot \MUZERO(I) + \gamma \cdot \MUONE(I) + \AT + (\gamma - 1) \cdot \sum_{j=1}^k x_{i_j} \cdot y_{i_j} \cdot \hat{x}_{i_j}. 
\end{equation*}
Finally, since $\gamma \leqslant 1$, we get that
\begin{equation*}
\OPT_{\BDIS{t}}(I) \leqslant \alpha \cdot \ALG_{\BDIS{t}}(I) + \beta \cdot \MUZERO(I) + \gamma \cdot \MUONE(I) + \AT.
\end{equation*}
Hence, assuming that~\ref{item:part_1_bdis_t_to_sch_t} and~\ref{item:part_2_bdis_t_to_sch_t} are satisfied, the above shows that if $\ALG_{\SCH{t}}$ is an $\ABC$-competitive Pareto-optimal algorithm for $\SCH{t}$ with respect to $\PAIRMU$, then $\ALG_{\BDIS{t}}$ is an $\ABC$-competitive algorithm for $\BDIS{t}$ with respect to $\PAIRMU$.

\textbf{Verifying~\ref{item:part_1_bdis_t_to_sch_t} and~\ref{item:part_2_bdis_t_to_sch_t}:}
Observe that $\ALG_{\BDIS{t}}$ and $\ALG_{\SCH{t}}$ gains the same profit on $I$ and $I'$. 
In particular, for each vertex on level $1$ $\ALG_{\BDIS{t}}$ always answers the same for $v_{i_j}$ as $\ALG_{\SCH{t}}$ does for $r_{i_j}$.
Moreover, $\ALG_{\SCH{t}}$ cannot gain any further profit on the requests that are non-challenge requests without creating an infeasible solution, since the only non-challenge requests are the intervals $r_{i_j,l}$, all of which overlap intervals that $\ALG_{\SCH{t}}$ has previously accepted.
Hence, for any $I \in \INSTANCES{\BDIS{t}}$, we have that $\ALG_{\SCH{t}}(\ORTRANS{\rho}(\ALG_{\SCH{t}},I)) = \ALG_{\BDIS{t}}(I)$ meaning that Inequality~\ref{item:part_1_bdis_t_to_sch_t} is satisfied.

To see that Inequality~\ref{item:part_2_bdis_t_to_sch_t} is satisfied we make the following observation.
For each vertex $v_i$ with $\ell(v_i) = 2$ and $x_i = 0$ there exists a vertex, $v_{i_j}$, on level $1$ such that $(v_i,v_{i_j}) \in E$, $y_{i_j} = 0$, and $x_{i_j} = 1$.
Indeed, a vertex $v_i$ only satisfies that $\ell(v_i) = 2$ if there exists some $j \in \{1,2,\ldots,k\}$ such that $i_j < i$, $y_{i_j} = 0$, and $(v_i,v_{i_j}) \in E$. 
Furthermore, since $x_i = 0$ and $(v_i,v_{i_j}) \in E$ it follows that $x_{i_j} = 1$, as otherwise $x$ encodes an infeasible solution. 

Hence, for any vertex $v_i$ with $\ell(v_i) = 2$ and $x_i = 0$, we can account the profit that $\OPT_{\BDIS{t}}$ gains from $v_i$ to the vertex $v_{i_j}$ with the above properties and smallest index, if there are more such vertices.
Any vertex $v_i$ with $\ell(v_i) = 1$ and $x_i = 0$, we account the profit that $\OPT_{\BDIS{t}}$ gains from $v_i$ to $v_i$ itself.
Hence, we may bound the profit of $\OPT_{\BDIS{t}}(I)$ by
\begin{align*}
\OPT_{\BDIS{t}}(I) \leqslant \sum_{j=1}^k \left(t \cdot x_{i_j} \cdot (1-y_{i_j}) + (1-x_{i_j})\right).
\end{align*}
At the same time, $\OPT_{\SCH{t}}$ gains profit $1$ from all intervals $r_{i_j}$ with $x_{i_j}' = 0$, and, due to the existence of the intervals $r_{i_j,l}$, $\OPT_{\SCH{t}}$ gains profit $t$ from each interval $r_{i_j}$ with $x_{i_j}' = 1$, in which case $y_{i_j} = 0$ by Equation~\eqref{eq:true_bits_challenge_requests_bdis_t_to_sch_t}.
Since we set $x_{i_j}' = 0$ whenever $x_{i_j} = y_{i_j} = 1$ and set $x_{i_j}' = x_{i_j}$ otherwise, we get that
\begin{align*}
\OPT_{\SCH{t}}(I') &= \sum_{j=1}^k \left( t \cdot x_{i_j}' \cdot (1-y_{i_j}) + (1-x_{i_j}') \right) = \sum_{j=1}^k \left( t \cdot x_{i_j} \cdot (1-y_{i_j}) + (1-x_{i_j}) + x_{i_j} \cdot y_{i_j} \right). 
\end{align*}
Putting it all together, we get that
\begin{align*}
\OPT_{\BDIS{t}}(I) \leqslant \; &\sum_{j=1}^k \left(t \cdot x_{i_j} \cdot (1-y_{i_j}) + (1-x_{i_j})\right) \\
\leqslant \; &\sum_{j=1}^k \left(t \cdot x_{i_j} \cdot (1-y_{i_j}) + (1-x_{i_j}) + x_{i_j} \cdot y_{i_j} - x_{i_j}\cdot y_{i_j} \cdot \hat{x}_{i_j} \right) \\
= \; &\sum_{j=1}^k \left(t \cdot x_{i_j} \cdot (1-y_{i_j}) + (1-x_{i_j}) + x_{i_j} \cdot y_{i_j}\right) - \sum_{j=1}^k x_{i_j}\cdot y_{i_j} \cdot \hat{x}_{i_j} \\
= \; &\OPT_{\SCH{t}}(I') - \sum_{j=1}^k x_{i_j} \cdot y_{i_j} \cdot \hat{x}_{i_j},
\end{align*}
implying that Inequality~\ref{item:part_2_bdis_t_to_sch_t} is satisfied, which finishes the proof of Item~\ref{item:pareto_optimal_reduction_new}.

\textbf{Towards~\ref{item:extend_beyond_pareto_optimality_new}: Proving that $\boldsymbol{\SCH{t}}$ is as hard as $\boldsymbol{\BDIS{t}}$.}
Let $\ALG_{\SCH{t}} \in \ALGS{\SCH{t}}$ be an $\ABC$-competitive algorithm for $\SCH{t}$.
We need to show that there exists an $\ABC$-competitive algorithm for $\BDIS{t}$.
If $\ALG_{\SCH{t}}$ is Pareto-optimal, then $\ALG_{\BDIS{t}} = \ORALG{\rho}(\ALG_{\SCH{t}})$ is an $\ABC$-competitive algorithm for $\BDIS{t}$ by Item~\ref{item:pareto_optimal_reduction_new}.
If, on the other hand, $\ALG_{\SCH{t}}$ is not Pareto-optimal, then there exists some $(\alpha',\beta',\gamma')$-competitive Pareto-optimal algorithm for $\SCH{t}$, say $\PAR_{\SCH{t}}$, with $\alpha' \leqslant \alpha$, $\beta' \leqslant \beta$, and $\gamma'\leqslant \gamma$.
Since $\PAR_{\SCH{t}}$ is Pareto-optimal, then $\ORALG{\rho}(\PAR_{\SCH{t}})$ is an $(\alpha',\beta',\gamma')$-competitive algorithm for $\BDIS{t}$, by Item~\ref{item:pareto_optimal_reduction_new}.
Using that $\alpha' \leqslant \alpha$, $\beta' \leqslant \beta$, $\gamma' \leqslant \gamma$, and the $(\alpha',\beta',\gamma')$-competitiveness of $\ORALG{\rho}(\PAR_{\SCH{t}})$, we get that
\begin{align*}
\OPT_{\BDIS{t}}(I) &\leqslant \alpha' \cdot \ORALG{\rho}(\PAR_{\SCH{t}}) + \beta' \cdot \MUZERO(I) + \gamma' \cdot \MUONE(I) + \AT \\
&\leqslant \alpha \cdot \ORALG{\rho}(\PAR_{\SCH{t}}) + \beta \cdot \MUZERO(I) + \gamma \cdot \MUONE(I) + \AT,
\end{align*}
which means that $\ORALG{\rho}(\PAR_{\SCH{t}})$ is an $\ABC$-competitive algorithm for $\BDIS{t}$.
\end{proof}

\subsection{Proving Lemma~\ref{lem:more_colors_is_easier}}\label{sec:proof_of_lemma_more_colors_is_easier}

In this section, we give the proof of Lemma~\ref{lem:more_colors_is_easier}. 
Similarly to Previous results, we need to prove a technical Lemma first.

\vspace{0.2cm}
\begin{lemma}\label{lem:pareto-optimal_for_bdist}
Any $\ABC$-competitive Pareto-optimal algorithm for $\BDIS{t}$ with respect to $\PAIRMU$ has $\gamma \leqslant 1$.
\end{lemma}
\begin{proof}
Assume towards contradiction that $\ALG_{\BDIS{t}}$ is an $\ABC$-competitive Pareto-optimal algorithm for $\BDIS{t}$ with $\gamma > 1$.
Since $\ASG{t}$ is as hard as $\BDIS{t}$ by Lemma~\ref{lem:reduction_bdis_t_to_asg_t}, the existence of $\ALG_{\BDIS{t}}$ implies the existence of an $\ABC$-competitive algorithm for $\ASG{t}$, say $\ALG_{\ASG{t}}$.
Observe by Lemma~\ref{lem:asg_t_to_bdis_t} that $\ALG_{\ASG{t}}$ is necessarily Pareto-optimal for $\ASG{t}$ as well.
Indeed, if $\ALG_{\ASG{t}}$ is not Pareto-optimal for $\ASG{t}$, then there exists some $(\alpha',\beta',\gamma')$-competitive Pareto-optimal algorithm for $\ASG{t}$, say $\PAR_{\ASG{t}}$, with $\alpha' \leqslant \alpha$, $\beta' \leqslant \beta$, and $\gamma' \leqslant \gamma$, where at least one of the three inequalities being strict. 
Without loss of generality, assume that $\alpha < \alpha'$.
Then, by Lemma~\ref{lem:asg_t_to_bdis_t}, there exists an $(\alpha',\beta',\gamma')$-competitive algorithm for $\BDIS{t}$ contradicting the Pareto-optimality of $\ALG_{\BDIS{t}}$.

Now, since $\ALG_{\ASG{t}}$ is an $\ABC$-competitive Pareto-optimal algorithm for $\ASG{t}$ with $\gamma > 1$, we have a contradiction with Theorem~\ref{thm:pareto_optimal_algorithms_for_asg_t}.
Thus, if $\ALG_{\BDIS{t}}$ is Pareto-optimal for $\BDIS{t}$ it cannot have $\gamma > 1$.
\end{proof}

With this technical lemma proven, we can prove Lemma~\ref{lem:more_colors_is_easier}.

\vspace{0.2cm}
\morecolorsiseasier*
\begin{proof}
Throughout the proof, we only measure prediction error using the canonical error measures $\PAIRMU$ defined in Equation~\eqref{eq:pairmu}. 
Therefore, we suppress notation and nomenclature related to error measures whenever possible.

We define a tuple $\rho = \ONLRED{\rho}$ that closely resembles a max-reduction from $\MCS{k}{kt}$ to $\BDIS{t}$, without actually being one.
Therefore, similarly to previous proofs, we split this proof into several parts:
\begin{enumerate}[label = {(\roman*)}]
\item \label{item:define_maps_new_new} Define two maps $\ORALG{\rho} \colon \ALGS{\BDIS{t}} \rightarrow \ALGS{\MCS{k}{kt}}$ and $\ORTRANS{\rho} \colon \ALGS{\BDIS{t}} \times \INSTANCES{\MCS{k}{kt}} \rightarrow \INSTANCES{\BDIS{t}}$.
\item \label{item:pareto_optimal_reduction_new_new} We show that if $\ALG_{\BDIS{t}} \in \ALGS{\BDIS{t}}$ is an $\ABC$-competitive Pareto-optimal algorithm for $\BDIS{t}$, then $\ALG_{\MCS{k}{kt}} = \ORALG{\rho}(\ALG_{\BDIS{t}})$ is an $\ABC$-competitive for $\MCS{k}{kt}$.
\item \label{item:extend_beyond_pareto_optimality_new_new} We show that $\BDIS{t}$ is as hard as $\MCS{k}{kt}$.
\end{enumerate}

\textbf{Towards~\ref{item:define_maps_new_new}: Definition of $\boldsymbol{{\ORALG{\rho}}}$ and $\boldsymbol{{\ORTRANS{\rho}}}$.}
Let $\ALG_{\BDIS{t}} \in \ALGS{\BDIS{t}}$ and $I = (x,\hat{x},r) \in \INSTANCES{\BDIS{t}}$.
We define $I' = (x',\hat{x}',r') = \ORTRANS{\rho}(\ALG_{\BDIS{t}},I)$ and $\ALG_{\MCS{k}{kt}} = \ORALG{\rho}(\ALG_{\BDIS{t}})$ as follows. 

When $\ALG_{\MCS{k}{kt}}$ receives a vertex, $v_i$, and a set of edges, $E_i$, with predicted bit $\hat{x}_i$ it determines the level of $v_i$, denoted $\ell(v_i)$, as follows.
Let $V_i = \{v_1,v_2,\ldots,v_i\}$ be the first $i$ vertices that has been revealed to $\ALG_{\MCS{k}{kt}}$.
If there exists $l \in \{1,2,\ldots,k\}$ such that for all vertices $u \in V_i$ on level $l$, either $u$ was rejected by $\ALG_{\MCS{k}{kt}}$ (i.e.\ $y_u = 1$) or $(u,v_i) \not\in E_i$, we let $\ell(v_i)$ be the smallest such $l$.
Otherwise, we let $\ell(v_i) = k+1$.
Then, if $\ell(v_i) \neq k+1$, $\ALG_{\MCS{k}{kt}}$ gives a challenge request to $\ALG_{\BDIS{t}}$ containing a new isolated vertex, $v_i'$, with predicted bit $\hat{x}_i$, and then $\ALG_{\MCS{k}{kt}}$ outputs the same for $v_i$ as $\ALG_{\BDIS{t}}$ does for $v_i'$, meaning that $\ALG_{\MCS{k}{kt}}$ accepts $v_i$ into its solution if and only if $\ALG_{\BDIS{t}}$ accepts $v_i'$ into its independent set.
If, on the other hand, $\ell(v_i) = k+1$, $\ALG_{\MCS{k}{kt}}$ rejects $v_i$.

Let $V^l = \{v_{l,i_1},v_{l,i_2},\ldots,v_{l,i_{s_l}}\}$ be the vertices on level $l$, for $l \in \{1,2,\ldots,k\}$, when $\ALG_{\MCS{k}{kt}}$ receives no more requests. 
Observe that there has been given a challenge request to $\ALG_{\BDIS{t}}$ for each vertex in $\bigcup_{l=1}^k V^l$.
We let $v_{l,i_j}'$ denote the challenge request given to $\ALG_{\BDIS{t}}$, corresponding to the request $v_{l,i_j}$ given to $\ALG_{\MCS{k}{kt}}$.
Since $\ALG_{\MCS{k}{kt}}$ gives the same prediction for $v_{l,i_j}'$ as it received for $v_{l,i_j}$, we have that $\hat{x}_{l,i_j}' = \hat{x}_{l,i_j}$ for all $l \in \{1,2,\ldots,k\}$ and all $j \in \{1,2,\ldots,s_l\}$. 
Next, we determine the true bits of the challenge requests.
For any $l \in \{1,2,\ldots,k\}$ and any $j \in \{1,2,\ldots,s_l\}$ we let
\begin{align}\label{eq:true_bits_mcs_1_t_to_mcs_k_kt}
x_{l,i_j}' = \begin{cases}
1-x_{l,i_j}, &\mbox{if $x_{l,i_j} = y_{l,i_j} = 1$} \\
x_{l,i_j}, &\mbox{otherwise.}
\end{cases}
\end{align}
We only have that $x_{l,i_j} \neq x_{l,i_j}'$ if $\ALG_{\MCS{k}{kt}}$ correctly identifies a true $1$ on one of the first $k$ levels.

After computing the true bits of the challenge requests we give a number of extra requests to $\ALG_{\BDIS{t}}$.
For each $l = 1,2,\ldots,k$ and each $j=1,2,\ldots,s_l$, if $x_{l,i_j}' = 1$ and $y_{l,i_j} = 0$, we give a \emph{block} of requests.
Each of these blocks contain $t$ new vertices $b_{l,i_j,p}'$, for $p = 1,2,\ldots,t$, each of which revealed together with the edge $(b_{l,i_j,p}',v_{l,i_j}')$ with true and predicted bits $x_{l,i_j,p}' = \hat{x}_{l,i_j,p}' = 0$.
Observe that $\OPT_{\BDIS{t}}$ accepts all $b_{l,i_j,p}'$ as it rejects $v_{l,i_j}'$, but since $\ALG_{\BDIS{t}}$ has accepted $v_{l,i_j}'$ it must reject all $b_{l,i_j,p}'$ to create a feasible solution.

\textbf{Showing that $\boldsymbol{{I' \in \INSTANCES{\BDIS{t}}}}$:}
We argue that $x'$ encodes an optimal solution to $I'$.
To this end, observe that $\OPT_{\BDIS{t}}$ accepts the vertices of all challenge requests with $x_{l,i_j}' = 0$.
The only way a challenge request may have true bit $x_{l,i_j}' = 1$ is if $y_{l,i_j} = 0$, in which case $v_{l,i_j}'$ is adjacent to the $t$ vertices $b_{l,i_j,p}'$, and so it is optimal to reject $v_{l,i_j}'$ and instead accept the $t$ vertices $b_{l,i_j,p}'$.
Lastly, all vertices in $I'$ have degree at most $t$, and so $I' \in \INSTANCES{\BDIS{t}}$.

\textbf{Bounding the error of $\boldsymbol{{I'}}$:}
By Equation~\eqref{eq:true_bits_mcs_1_t_to_mcs_k_kt} we find that the error of $I'$ may be different from the error of $I$.
We bound the error of $I'$ by a function of the error of $I$ as follows.
The only time we have that $x_{l,i_j}' \neq x_{l,i_j}$ is when $x_{l,i_j} = y_{l,i_j} = 1$, in which case we set $x_{l,i_j}' = 0$.
Hence, if $\hat{x}_{l,i_j} = \hat{x}_{l,i_j}' = 0$ then we change an incorrect prediction into a correct prediction which will make the value of $\MUZERO$ decrease by $1$.
On the other hand, if $\hat{x}_{l,i_j} = \hat{x}_{l,i_j}' = 1$ then we change a correct prediction to an incorrect prediction which will make the value of $\MUONE$ increase by one.
Therefore,
\begin{align}
\MUZERO(I') &\leqslant \MUZERO(I) - \sum_{l=1}^k \sum_{j=1}^{s_l} (1-\hat{x}_{l,i_j}) \cdot x_{l,i_j} \cdot y_{l,i_j} \leqslant \MUZERO(I) \label{eq:M0''''} \\
\MUONE(I') &\leqslant \MUONE(I) + \sum_{l=1}^k \sum_{j=1}^{s_l} \hat{x}_{l,i_j} \cdot x_{l,i_j} \cdot y_{l,i_j}\label{eq:M1''''}
\end{align}
These bounds are similar to Equations~\eqref{eq:bdis_t_to_sch_t_muzero_change} and~\eqref{eq:bdis_t_to_sch_t_muone_change} with the only exception that we have more levels this time.

\textbf{Towards~\ref{item:pareto_optimal_reduction_new_new}: We show that $\boldsymbol{\ABC}$-competitive Pareto-optimal algorithms for $\BDIS{t}$ translates into $\boldsymbol{\ABC}$-competitive algorithms for $\MCS{k}{kt}$.}
Let $\ALG_{\BDIS{t}} \in \ALGS{\BDIS{t}}$ be an $\ABC$-competitive Pareto-optimal algorithm for $\BDIS{t}$.
In the following, we show that $\ALG_{\MCS{k}{kt}} = \ORALG{\rho}(\ALG_{\BDIS{t}})$ is an $\ABC$-competitive algorithm for $\MCS{k}{kt}$, where $\ORALG{\rho}$ is the map defined in Item~\ref{item:define_maps_new_new}.
To this end, recall from Lemma~\ref{lem:pareto-optimal_for_bdist} that $\gamma \leqslant 1$.

\textbf{Structure of the analysis:}
With the above setup, we show that
\begin{enumerate}[label = {(\Alph*)}]
\item $\ALG_{\BDIS{t}}(I') \leqslant \ALG_{\MCS{k}{kt}}(I)$, and \label{item:condition_alg_mcs_k_kt_to_mcs_1_t}
\item $\OPT_{\MCS{k}{kt}}(I) \leqslant \OPT_{\BDIS{t}}(I') - \sum_{l=1}^k \sum_{j=1}^{s_l} x_{l,i_j} \cdot y_{l,i_j} \cdot \hat{x}_{l,i_j}$. \label{item:condition_opt_mcs_k_kt_to_mcs_1_t}
\end{enumerate}
We show that these two inequalities imply that $\BDIS{t}$ is as hard as $\MCS{k}{kt}$.
Hence, assume that~\ref{item:condition_alg_mcs_k_kt_to_mcs_1_t} and~\ref{item:condition_opt_mcs_k_kt_to_mcs_1_t} are satisfied and that $\ALG_{\BDIS{t}}$ is an $\ABC$-competitive Pareto-optimal algorithm for $\BDIS{t}$ with additive constant $\AT$.
Then, by~\ref{item:condition_opt_mcs_k_kt_to_mcs_1_t},
\begin{equation*}
\OPT_{\MCS{k}{kt}}(I) \leqslant \OPT_{\BDIS{t}}(I') - \sum_{l=1}^k \sum_{j=1}^{s_l} x_{l,i_j} \cdot y_{l,i_j} \cdot \hat{x}_{l,i_j}.
\end{equation*}
Then, by the $\ABC$-competitiveness of $\ALG_{\BDIS{t}}$, we have that
\begin{equation*}
\OPT_{\MCS{k}{kt}}(I) \leqslant \alpha \cdot \ALG_{\BDIS{t}}(I') + \beta \cdot \MUZERO(I') + \gamma \cdot \MUONE(I') + \AT - \sum_{l=1}^k \sum_{j=1}^{s_l} x_{l,i_j} \cdot y_{l,i_j} \cdot \hat{x}_{l,i_j}.
\end{equation*}
Then, by~\ref{item:condition_alg_mcs_k_kt_to_mcs_1_t} and Equations~\eqref{eq:M0''''} and~\eqref{eq:M1''''},
\begin{align*}
\OPT_{\MCS{k}{kt}}(I) \leqslant \; &\alpha \cdot \ALG_{\MCS{k}{kt}}(I) + \beta \cdot \MUZERO(I) + \gamma \cdot \MUONE(I) \\
&+ \AT + (\gamma - 1) \cdot \sum_{l=1}^k \sum_{j=1}^{s_l} x_{l,i_j} \cdot y_{l,i_j} \cdot \hat{x}_{l,i_j}.
\end{align*}
Finally, since $\ALG_{\BDIS{t}}$ is Pareto-optimal, then, by Lemma~\ref{lem:pareto-optimal_for_bdist}, $\gamma \leqslant 1$ and so
\begin{equation*}
\OPT_{\MCS{k}{kt}}(I) \leqslant \alpha \cdot \ALG_{\MCS{k}{kt}}(I) + \beta \cdot \MUZERO(I) + \gamma \cdot \MUONE(I) + \AT.
\end{equation*}
Hence, assuming that~\ref{item:condition_alg_mcs_k_kt_to_mcs_1_t} and~\ref{item:condition_opt_mcs_k_kt_to_mcs_1_t} are satisfied, the above translates the $\ABC$-competitiveness of $\ALG_{\BDIS{t}}$ to $\ABC$-competitiveness of $\ALG_{\MCS{k}{kt}}$ with respect to $\PAIRMU$.


\textbf{Verifying~\ref{item:condition_alg_mcs_k_kt_to_mcs_1_t} and~\ref{item:condition_opt_mcs_k_kt_to_mcs_1_t}:}
With the above setup, $\ALG_{\MCS{k}{kt}}$ always creates a feasible solution (vertices on level $l$ can be colored by the same color, and all vertices on level $k+1$ are not given a color).
Further, $\ALG_{\BDIS{t}}$ can only gain profit from the vertices $v_{l,i_j}'$, for $l \in \{1,2,\ldots,k\}$ and $j=1,2,\ldots,s_l$, as otherwise it creates an infeasible solution.
Hence, for any $I \in \INSTANCES{\MCS{k}{kt}}$,
\begin{align*}
\ALG_{\MCS{k}{kt}}(I) = \ALG_{\BDIS{t}}(I') = \sum_{l=1}^k \sum_{j=1}^{s_l} (1-y_{l,i_j}).
\end{align*}
This verifies~\ref{item:condition_alg_mcs_k_kt_to_mcs_1_t}.

Towards~\ref{item:condition_opt_mcs_k_kt_to_mcs_1_t}, observe that $\OPT_{\BDIS{t}}$ gains profit $1$ for each request in $I'$ with true bit $0$.
We express the profit of $\OPT_{\BDIS{t}}$ as a function of the true bits of the vertices on the first $k$ levels from $I$.
To this end, let $v_{l,i_j}$ be a vertex that was requested for $\MCS{k}{kt}$, with $l \in \{1,2,\ldots,k\}$.
If $x_{l,i_j} = 0$, then $\OPT_{\BDIS{t}}$ gains profit $1$, since the associated vertex $c_{l,i_j}$ also has true bit $0$.
If, on the other hand $x_{l,i_j} = 1$, then $\OPT_{\BDIS{t}}$ gains profit $1$ if $y_{l,i_j} = 1$ as we set $x_{l,i_j}' = 0$ whenever $x_{l,i_j} = y_{l,i_j} = 1$, and $\OPT_{\BDIS{t}}$ gains profit $t$ if $y_{l,i_j} = 0$ due to the block of request given to $\ALG_{\BDIS{t}}$ after $\ALG_{\MCS{k}{kt}}$ does not receive any further requests.
Therefore,
\begin{align}\label{eq:opt_mcs1t}
\OPT_{\BDIS{t}}(I') = \sum_{l=1}^k \sum_{j=1}^{s_l} \left( (1-x_{l,i_j}) + x_{l,i_j}\cdot y_{l,i_j} + t \cdot x_{l,i_j} \cdot (1-y_{l,i_j}) \right).
\end{align}

Next, we bound $\OPT_{\MCS{k}{kt}}$.
Recall that $\OPT_{\MCS{k}{kt}}(I) = \sum_{i=1}^n (1-x_i)$.
However, to make $\OPT_{\MCS{k}{kt}}$ more easy to relate to $\OPT_{\BDIS{t}}$, we bound the profit of $\OPT_{\MCS{k}{kt}}$ as a function of the true bits of the vertices on the first $k$ levels from $I$.

To this end, fix any vertex $v_i$ for which $x_i = 0$ so $\OPT_{\MCS{k}{kt}}$ gains profit from $v_i$.
If $\ell(v_i) \in \{1,2,\ldots,k\}$ then there exist $l \in \{1,2,\ldots,k\}$ and $j \in \{1,2,\ldots,s_l\}$ such that $v_i = v_{l,i_j}$ and so we may assign the profit that $\OPT_{\MCS{k}{kt}}$ gains from $v_i$ to $x_{l,i_j}$.
On the other hand, if $\ell(v_i) = k+1$ then for all $l \in \{1,2,\ldots,k\}$, there exists a vertex $u_l$ with $\ell(u_l) = l$, such that $(u_l,v_i) \in E$, $y_{u_l} = 0$, and $x_{u_l} = 1$, as otherwise $v_i$ would have been placed on another level.
We assign a profit of $\frac{1}{k}$ to each vertex on levels $1,2,\ldots,k$ that is adjacent to $v_i$.
Since $\MAXDEGREE(G) \leqslant kt$, the profit of a vertex, $v$, on level $l \in \{1,2,\ldots,k\}$ with $x_v = 1$ and $y_v = 0$ is at most $\frac{1}{k} \cdot kt = t$. 
With this accounting scheme, we can upper bound the profit of $\OPT_{\MCS{k}{kt}}$ by
\begin{align}\label{eq:opt_mcskkt}
\OPT_{\MCS{k}{kt}}(I) \leqslant \sum_{l=1}^k \sum_{j=1}^{s_l} \left( (1-x_{l,i_j}) + t \cdot x_{l,i_j}\cdot (1-y_{l,i_j}) \right).
\end{align}
Combining the bounds from Equations~\eqref{eq:opt_mcs1t} and~\eqref{eq:opt_mcskkt}, we get that
\begin{align*}
\OPT_{\MCS{k}{kt}}(I) \overset{\eqref{eq:opt_mcskkt}}{\leqslant} \; &\sum_{l=1}^k \sum_{j=1}^{s_l} \left( (1-x_{l,i_j}) + t \cdot x_{l,i_j}\cdot (1-y_{l,i_j}) \right) \\
= \; &\sum_{l=1}^k \sum_{j=1}^{s_l} \left( (1-x_{l,i_j}) + x_{l,i_j} \cdot y_{l,i_j} + t \cdot x_{l,i_j}\cdot (1-y_{l,i_j}) \right) \\
&- \sum_{l=1}^k \sum_{j=1}^{s_l} x_{l,i_j} \cdot y_{l,i_j} \\
\overset{\eqref{eq:opt_mcs1t}}{=} \; &\OPT_{\BDIS{t}}(I') - \sum_{l=1}^k \sum_{j=1}^{s_l} x_{l,i_j} \cdot y_{l,i_j} \\
\leqslant \; &\OPT_{\BDIS{t}}(I') - \sum_{l=1}^k \sum_{j=1}^{s_l} x_{l,i_j} \cdot y_{l,i_j} \cdot \hat{x}_{l,i_j}.
\end{align*}
This verifies~\ref{item:condition_opt_mcs_k_kt_to_mcs_1_t}, which finishes the proof of Item~\ref{item:pareto_optimal_reduction_new_new}.

\textbf{Towards~\ref{item:extend_beyond_pareto_optimality_new_new}: Proving that $\BDIS{t}$ is as hard as $\MCS{k}{kt}$.}
This proof is very similar to that of Item~\ref{item:extend_beyond_pareto_optimality_new} in the proof of Lemma~\ref{lem:bdis_t_to_sch_t}.
\end{proof}

\section{Algorithmic and Complexity Theoretical Consequences}\label{sec:consequences}

In this section we relate the results from Sections~\ref{sec:a_collection_of_problems_from_CCWM_max} and~\ref{sec:independent_set_vs_asg_t} to the complexity theory introduced in~\cite{BBFL25}, and in particular the complexity classes $\CCWM{\PAIRMUFORCC}{t}$.
As proven in~\cite{BBFL25}, proving membership and hardness results for $\CCWM{\PAIRMUFORCC}{t}$ imply several algorithmic results.

All results in this section are direct consequences of Theorems~\ref{thm:section_3} and~\ref{thm:main_result}.

\vspace{0.2cm}
\begin{theorem}\label{thm:asg_t_is_as_hard_as_bdis_t_purely_online}
For all $t \in \ZZ^+$, $\BDIS{t}$, $\SCH{t}$, $\SP{t}$, and $\CLI{t}$ are $\CCWM{\PAIRMUFORCC}{t}$-complete.
\end{theorem}

\begin{remark}
Theorem~\ref{thm:asg_t_is_as_hard_as_bdis_t_purely_online} also holds in the purely online setting, i.e.\ with respect to $\PAIRZM$. 
\end{remark}

Furthermore, we get several membership and hardness results for $\CCWM{\PAIRMUFORCC}{t}$:

\vspace{0.2cm}
\begin{theorem}
For all $k,t \in \ZZ^+$,
\begin{itemize}
\item $\MCS{k}{kt} , \MAT{\lfloor t/2 \rfloor + 1} \in \CCWM{\PAIRMUFORCC}{t}$, and
\item $\SCHEDULING$, $\IS$, $\CLIQUE$, $\SETPACKING$ are $\CCWM{\PAIRMUFORCC}{t}$-hard.
\end{itemize}
\end{theorem}

Proving these connections to $\CCWM{\PAIRMUFORCC}{t}$ yields a number of positive and negative algorithmic results for the problems considered:

\vspace{0.2cm}
\begin{corollary}\label{cor:strong_lower_bounds}
Let $t \in \ZZ^+ \cup \{\infty\}$, let $P \in \{\BDIS{t},\SCH{t},\SP{t},\CLI{t}\}$, and let $\ALG$ be an $\ABC$-competitive algorithm for $P$ with respect to $\PAIRMU$.
Then,
\begin{enumerate}[label = {(\roman*)}]
\item $\alpha + \beta \geqslant t$, \label{item:alpha+beta_geq_t}
\item $\alpha + (t-1) \cdot \gamma \geqslant t$, and \label{item:alpha+(t-1)gamma_geq_t}
\item if $\beta = \gamma = 0$ then $\alpha \geqslant t$. \label{item:no_t-eps_comp_alg_with_preds}
\end{enumerate}
\end{corollary}
\begin{proof}
This is a direct consequence of a more general statement from~\cite{BBFL25}.
\end{proof}

\vspace{0.2cm}
\begin{corollary}\label{cor:positive_result_members}
Let $P \in \{\BDIS{t},\SP{t},\SCH{t},\CLI{t},\MCS{k}{kt},\MAT{\lfloor t/2 \rfloor + 1}\}$ for $k,t \in \ZZ^+$.
Then, for all $1 \leqslant \alpha \leqslant t$, there exists an $(\alpha,t-\alpha,1)$-competitive and a $(t,0,0)$-competitive algorithm for $P$.
If $P \in \{\BDIS{t},\SP{t},\SCH{t},\CLI{t}\}$, these algorithms are Pareto-optimal.
\end{corollary}

\section{Concluding Remarks}

We have shown several reductions implying results on the relative hardness of several online maximization and minimization problems. 
In particular, we have shown that there exists an $\ABC$-competitive algorithm for $\BDIS{t}$ if and only if there exists an $\ABC$-competitive algorithm for its dual minimization problem $\BDVC{t}$.
In full generality, our results imply that any two problems on the following list are equivalent in terms of hardness: $\ASG{t}$, $\BDVC{t}$, $\BDIS{t}$, $\SP{t}$, $\SCH{t}$, $\CLI{t}$, and \emph{Online $t$-Bounded Overlap Interval Rejection with Predictions} (see~\cite{BBFL25}).

There are still several open problems from~\cite{BBFL25} that remain open, which includes considering randomized algorithms instead of deterministic algorithms, and changing the hardness measure from competitiveness to e.g.\ random order or relative worst order.
Beyond these, it would be interesting to consider other prediction schemes beyond the binary prediction scheme.


\bibliographystyle{plain}
\bibliography{refs.bib}

\end{document}